\documentclass[reqno]{amsart}

\usepackage[margin=3cm]{geometry}

\def\dxplot{.96}
\usepackage{arXiv_style}

\usepackage{fontawesome}

\makeatletter
\def\blfootnote{\gdef\@thefnmark{}\@footnotetext}
\makeatother

\title[\CODSSI]{COD-\textsc{ssi}: Enforcing Mutual Privacy for Credential Oblivious Disclosure in Self Sovereign Identity}

\author[E.\ Onofri et al.]{}

\begin{document}

\blfootnote{$^{\star}$ Corresponding Author, (\href{mailto:elia.onofri@kaust.edu.sa}{\faEnvelopeO}) \texttt{elia.onofri@kaust.edu.sa}, (\href{https://www.eliaonofri.it}{\faGlobe}) \texttt{www.eliaonofri.it}}

\maketitle

\vspace{-1em}

\begin{center}
    \begin{minipage}{.89\linewidth}\centering
        \textsc{Elia Onofri}$^{\, \textsc{a},\, \star,\, \orcidlink{0000-0001-8391-2563}}$,
        \textsc{Andrea De Salve}$^{\, \textsc{b},\, \orcidlink{0000-0003-1691-7182}}$,
        \textsc{Paolo Mori}$^{\, \textsc{c},\, \orcidlink{0000-0002-6618-0388}}$,\\
        \textsc{Laura Emilia Maria Ricci}$^{\, \textsc{d},\, \orcidlink{0000-0002-8179-8215}}$,
        \textsc{Roberto Di Pietro}$^{\, \textsc{a},\, \orcidlink{0000-0003-1909-0336}}$.
        \\
        \bigskip
        \begin{minipage}{.9\linewidth}\centering
            \footnotesize
            $^\textsc{a}$Computer, Electrical and Mathematical Sciences and Engineering (CEMSE) Division,\\
            King Abdullah University of Science and Technology (KAUST)\\
            Thuwal 23955, Saudi Arabia
        \end{minipage}

        \medskip

        \begin{minipage}{.9\linewidth}\centering
            \footnotesize
            $^\textsc{b}$Institute of Applied Sciences and Intelligent Systems (ISASI),\\
            National Research Council of Italy (CNR),\\
            Lecce 73100, Italy
        \end{minipage}

        \medskip

        \begin{minipage}{.45\linewidth}\centering
            \footnotesize
            $^\textsc{c}$Istituto di Informatica e Telematica (IIT),\\
            National Research Council of Italy (CNR),\\
            Pisa 56127, Italy
        \end{minipage}
        \begin{minipage}{.45\linewidth}\centering
            \footnotesize
            $^\textsc{d}$Department of Computer Science,\\
            University of Pisa,\\
            Pisa 56127, Italy
        \end{minipage}
    \end{minipage}
\end{center}

\medskip
\thispagestyle{empty}

\begin{abstract}

    The Self-Sovereign Identity (SSI) paradigm is instrumental for decentralised identity management, allowing an entity to create, manage, and present their digital credentials without relying on centralised authorities.
    Credential selective disclosure is one of the most attractive privacy-preserving features of SSI, allowing users to reveal only the minimum necessary information from their credentials.
    However, current selective disclosure mechanisms primarily focus on protecting the privacy of credential Holders, while offering limited protection to the Verifiers of credentials.
    Indeed, the specific credential information requested by a Verifier can inadvertently reveal to credential Holders sensitive information, including internal decision-making criteria, business rules, or strategic plans.
    In this work, we address this threat by proposing, to the best of our knowledge, the first approach that enforces mutual privacy in credential exchanges.
    To this end, we introduce \CODSSI (Claim Oblivious Disclosure for SSI), a novel framework that leverages Oblivious Pseudorandom Functions to allow Verifiers to selectively access a subset of claims without revealing which specific claims were accessed to the credential Holder.
    The security of our solution is formally verified and its feasibility is assessed through the experimental evaluation of our open-source prototype implementation.
    These results show that provable mutual privacy in the context of SSI can be achieved with just moderate computational and communication overhead.
    
    \medskip
    
    
    \smallskip
    
    \noindent{\bf AMS-MSC 2020:} Code1 \sep Code2 \sep Code3.

\end{abstract}

\begin{multicols}{2}
    \small
    \tableofcontents
\end{multicols}

\newpage


\section{Introduction}

The Self Sovereign Identity (\SSI) \cite{tobin2016inevitable} paradigm has been recently introduced to give back to users the control of their identities and of the data paired to them, \ie, the claims describing users' features, such as, \eg, their degree attestations, their vaccination certificates, or their proofs of income.
In fact, in traditional centralised identity models, users' identities and all the related data are managed by centralised entities, called Identity Providers (\textsf{IdPs}), which release such data to third parties.
For instance, when you access a service using Google Login, your name, email address, and profile picture are always released to the service provider~\cite{gLogin}.
Instead, adopting the \SSI model, each user creates their own Decentralised IDentity (\DID) \cite{didw3c}, typically leveraging the blockchain technology, and manages their own wallet where their identity-related data are stored in Verifiable Credentials (\VCs) \cite{vcw3c} form.
\VCs are issued by entities that are trusted for certifying specific features (claims, in the \SSI jargon) of users.
For instance, a University may issue a \VC to one of its students attesting that the cited student has obtained a degree in Computer Science with full marks and honours in a given year.
When interacting with third parties, called verifiers, users can choose from their wallet the subset of their \VCs they want to disclose, in order to keep private the remaining ones.
Since a \VC could contain several claims, selective disclosure techniques \cite{SalveLMR22} can be applied to further increase the granularity of control of the user on their data, allowing them to disclose only a subset of the claims embedded in each \VC. 
For instance, in the previous example, the student could disclose only the kind of degree attestation they earned, \ie, a degree in Computer Science, without disclosing the marks or the date of graduation.

\begin{figure*}[t]
    \centering
    \begin{tikzpicture}[scale=1.2,
        overbrace/.style={ultra thick, decoration={brace, amplitude=15pt}, decorate},
        bracenode/.style={pos=0.5,anchor=north,yshift=1.2cm},
        font=\footnotesize
    ]
        \newcommand*\circled[1]{\scalebox{.8}{\tikz[baseline=(char.base)]{\node[black,shape=circle,draw,fill=greenM!50,inner sep=1.3pt] (char) {#1};}}}

        \node[label={above:\textsc{Issuer}}] (I) at (0, 0) {\includegraphics[width=1cm]{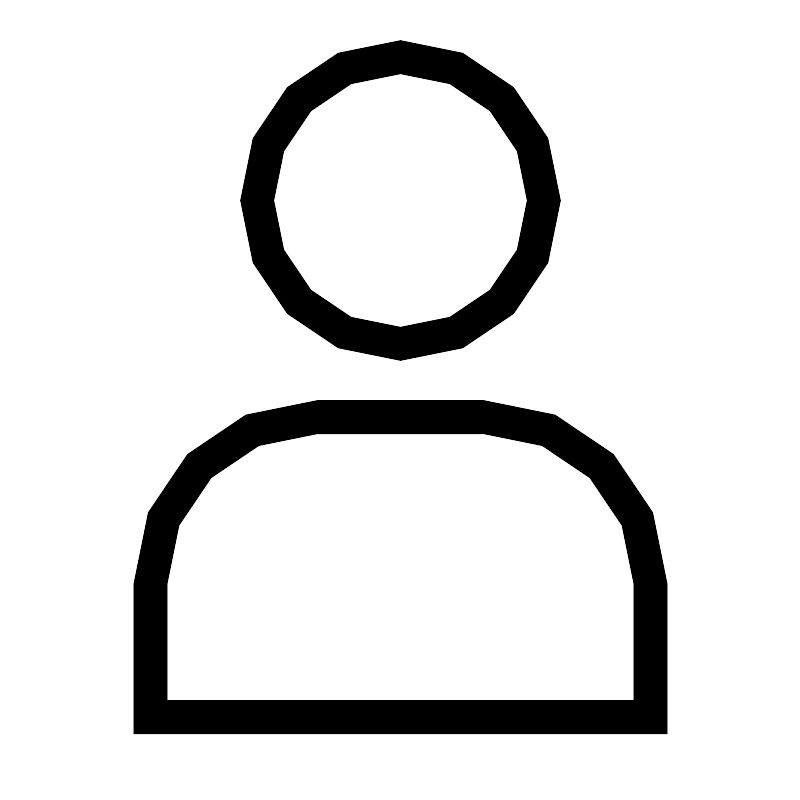}};
        \node[label={above:\textsc{Holder}}] (H) at (5, 0) {\includegraphics[width=1cm]{figs/user-black.png}};
        \node[label={above:\textsc{Verifier}}] (V) at (10, 0) {\includegraphics[width=1cm]{figs/user-black.png}};

        \draw[-latex] (0, -1.0) --node[above]{\VC} (5, -1.0);

        \draw[dashed] (-.7, -1.75) -- (6, -1.75);
        
        \draw[-latex] (5, -2.1) --node[above]{Send \VP} (10, -2.1);

        \draw[dashed] (4, -2.9) -- (11, -2.9);

        \draw[latex-, color=orangeM] (5, -3.65) --node[above=-.1em]{Request claim(s) data} (10, -3.65);
        \draw[-latex, color=orangeM] (5, -4.15) --node[above=-.1em]{Send claim(s) data} (10, -4.15);
        
        \draw[dashed] (4, -5.05) -- (11, -5.05);
        
        \draw[-latex, color=blueM] (5, -5.85) --node[above=-.1em]{Send Encrypted Claims data} (10, -5.85);
        \draw[latex-, color=blueM] (5, -6.35) --node[above=-.1em]{Request (oblivious) claim(s) key} (10, -6.35);
        \draw[-latex, color=blueM] (5, -6.85) --node[above=-.1em]{Send (oblivious) claim(s) key} (10, -6.85);

        \draw[-latex, double] (I) -- +(0, -2.2);
        \draw[-latex, double] (H) -- +(0, -7.35);
        \draw[-latex, double] (V) -- +(0, -8.0);

        \node[draw, fill=white] at (0, -1.0) {Issue \VC};
        \node[draw, fill=white] at (5, -1.4) {Verify \VC};
        \node[draw, fill=white] at (5, -2.1) {Create \VP};
        \node[draw, fill=white, align=center, anchor=east] at (10.3, -2.5) {Verify \VP};

        \node[orangeM, draw, fill=white, anchor=east] (CollClaims) at (5.3, -3.3) {Collect Claims};
        \draw[-latex, orangeM] (3.3, -1.0) |- (CollClaims);

        \node[blueM, draw, fill=white, anchor=east] (EncClaims) at (5.3, -5.5) {Encrypt Claims};
        \draw[-latex, blueM] (2.95, -1.0) --node[above,rotate=-90,black,xshift=-.85cm]{Claim data} (2.95, -5.5) --  (EncClaims);

        \node[draw, color=orangeM, fill=white, align=center, anchor=east] at (10.3, -4.65) {Verify Claim};
        \node[draw, color=blueM, fill=white, align=center, anchor=east] at (10.3, -7.35) {Decrypt \&\\Verify Claim};

        \node at (10.5,-1.9) {\scalebox{3}[7.5]{\}}};
        \node [rotate=90, align=center, anchor=north] at (10.65,-1.9) {Classical\\ Workflow};

        \node [color=orangeM] at (10.5,-3.9) {\scalebox{3}[7.6]{\}}};
        \node [rotate=90, align=center, color=orangeM, anchor=north] at (10.65,-3.9) {Selective\\ claim(s)\\ disclosure};

        \node [color=blueM] at (10.5,-6.3) {\scalebox{3}[9.8]{\}}};
        \node [rotate=90, align=center, color=blueM, anchor=north] at (10.65,-6.2) {Oblivious\\ claim(s)\\ disclosure};

        \node at (1,-1) {\circled{1}};
        \node at (6,-2.1) {\circled{2}};
        \node at (3.125,-1.5) {\circled{3}};
        \node at (9.5,-3.65) {\circled{4}};
        \node at (5.5,-4.15) {\circled{5}};
        \node at (5.5,-5.85) {\circled{6}};
        \node at (9.5,-6.35) {\circled{7}};
        \node at (5.5,-6.85) {\circled{8}};

        \node[black!50, draw, thick, anchor=north west, align=left, text width=4.25cm] at (-1.15, -2.5) {\scriptsize
            After a \VC is issued by the Issuer \circled{1}, the Holder creates a \VP and shares it with the Verifier \circled{2}.\\[.5em]
            In the classical approach (black, top), claims are embedded within the \VC  and hence received in the \VP, otherwise they are provided as separate claim data \circled{3}.\\[.5em]
            In the \emph{selective disclosure} setting (orange, middle), the Verifier explicitly requests \circled{4} and receives \circled{5} the selected plain claims.\\[.5em]
            In the \CODSSI approach (blue, bottom), claims are transmitted\- in encrypted form \circled{6}, and the Verifier engages in an OPRF-based protocol with the Holder to obliviously derive the corresponding decryption keys \circled{7}, \circled{8}.
        };
    \end{tikzpicture}
    \caption{
        End-to-end flow of credential issuance, presentation, and verification.
    }
    \label{fig:issuer-holder-verifier}
\end{figure*}

Summarising, the \SSI paradigm provides full support to users interacting with verifiers to preserve their privacy by disclosing only a subset of claims for a subset of the \VCs they own.
In some scenarios, however, also Verifiers could have their privacy requirements in the process of \SSI credentials exchange and verification, as they might want to keep secret from the subject the set of claims they need to know, which could even change from time to time -- in the following subsection we report a few examples where this feature proves its value. 
It is worth noticing that, to the best of our knowledge, the above highlighted gap is currently not addressed in the literature.

In this regard, this work seeks to overcome the aforementioned limitation by introducing \CODSSI, a Claim Oblivious Disclosure solution for Self Sovereign Identity.
\CODSSI enforces Verifiers' privacy by allowing them to selectively disclose a subset of the users' claims, while ensuring that the users remain unaware of (oblivious to) which specific claims have been actually disclosed.
In a nutshell: a user selects a subset of $N$ claims out of their available \VCs, \ie the claims that the user is willing to disclose, while restricting the maximum number of claims that can be actually revealed to $N_o$, with $N_o<N$.
The Verifier then chooses up to $N_o$ claims to be disclosed among the available ones, without the user knowing which specific claims have been actually disclosed.
Figure~\ref{fig:issuer-holder-verifier} provides the workflow for credential life management, emphasising the role of the selective disclosure and oblivious selective disclosure w.r.t.\ the classical \SSI setting.


\subsection{Motivating Examples} The proposed approach finds application in multiple use cases where Verifiers' privacy is relevant, such as the two described below.
\subsubsection{Financial risk assessment.}
A compelling application of \CODSSI arises in the domain of trade secrets, which are confidential information (such as process, formula, and algorithm) that provide a subject with economic and competitive advantages over competitors, and they can be protected by patent or copyright.
As for example, creditworthiness and loan evaluation provided by financial institutions are required to employ risk assessment models that rely on a limited number of applicant attributes to derive a credit score, such as a standard Employment Status credential and a Financial Status credential issued by the government.
Suppose the Verifier always requests from the Financial Status credential the income and one additional risk indicator: debt-ratio (if the applicant's employment sector is considered stable) or delinquency count (if instead the employment sector is considered high-risk).
If the Holder can observe which claim is requested (debt-ratio vs delinquency count), they can immediately infer the Verifier's internal classification of their employment profile. By slightly modifying the employment credential across sessions, an attacker could reverse-engineer the Verifier's decision logic or scoring thresholds.
More generally, disclosing the exact subset of attributes used in the computation could reveal the structure of the risk model, thereby undermining the confidentiality of proprietary scoring methodologies that institutions are expected to safeguard under prudential and supervisory frameworks (cf.\ CRR~\cite{eu2013crr} or Basel~III and its Pillar~3 disclosure framework~\cite{bcbs2018pillar3}).
Furthermore, GDPR \cite{GDPR} recognises that the right to the protection of personal data must not affect the fundamental rights and freedoms of others, including trade secrets (Recital 4 of the GDPR).
\CODSSI is a solution: the Holder makes available a set of attributes but releases only the ones actually requested by the Verifier for the risk assessment, while the Verifier's choice of which attributes to access remains hidden, enforcing both GDPR-aligned proportionality in personal data handling and further preserving the confidentiality of financial institutions' internal models.


\subsubsection{Medical data access.}
A further application can be identified in the management of access to Electronic Health Records (EHRs) for research purposes. 
Medical data are considered a \emph{special category of personal data} under the GDPR (\cite[Art.~9]{GDPR}), subject to strict requirements of proportionality and purpose limitation, while in the United States comparable obligations are imposed by the Health Insurance Portability and Accountability Act (HIPAA, \cite{HIPAA1996}). 
Researchers often require access to a limited set of medical attributes --for example, the results of specific laboratory tests, diagnostic codes, or prescribed treatments-- in order to validate hypotheses or train predictive models. 
However, disclosing \emph{which} subset of fields is accessed may in itself reveal the focus of the ongoing research, thus risking premature exposure of scientific work and potential misappropriation of intellectual contributions. 
\CODSSI provides a balanced solution: patients or custodians release no more than the attributes effectively queried among the available ones, 
while the researchers' choice of which fields to access remains hidden. 
Hence, the framework ensures compliance with GDPR's \emph{data minimisation principle} and \emph{data protection by design} (\cite[Arts.\ 5(1)(c), 25]{GDPR}), while simultaneously protecting the confidentiality of research agendas, thereby fostering innovation without compromising patient privacy or institutional compliance.


\subsection{Contributions}

This paper provides the following contributions:
\begin{itemize}
    \item We define \CODSSI, the first solution, to the best of our knowledge, to enforce Verifier privacy \SSI credentials exchange, disclosure and verification;
    \item We provide structured proofs demonstrating the security of our approach under an honest-but-curious client and a malicious server model and we discuss the extensibility to a model where both parties are malicious; 
    \item We develop a Proof of Concept (PoC) implementing the proposed approach to show its feasibility;
    \item We conduct a set of experiments on the PoC to assess the introduced overhead, showing that provable mutual privacy in \SSI can be achieved with just moderate computational and communication overhead; and,
    \item We will release the source code for the practitioners and the scientific community to build up on our findings.
\end{itemize}


\subsection{Roadmap}

The remainder of this work is structured as follows.  
Section~\ref{sec:relatedWork} surveys related literature on the core primitives underlying \CODSSI, namely selective disclosure, encryption in identity systems, and oblivious key derivation.  
These primitives are then formalised in Section~\ref{sec:preliminaries}, which lays the foundations for the \CODSSI model presented in Section~\ref{sec:COD-SSI}.
Section~\ref{sec:codssi-security} discusses the security of the construction, also individually analysing the role of malicious internal parties.
Section~\ref{sec:exp} describes our experimental setting and reports performance results, comparing the overhead of our oblivious construction against state-of-the-art (non-oblivious) selective disclosure approaches.
Section~\ref{sec:conclusions} concludes with a summary of contributions and directions for future work.
Appendix~\ref{app:results} provides the experimental results for the \VC creation (shared with classical selective disclosure).
For a detailed treatment of the security analysis of \CODSSI, the reader is referred to Appendix~\ref{app:security}.
A table summarising all the symbols is located in Appendix~\ref{app:nomenclature}.


\section{Related Work}
\label{sec:relatedWork}

Research on privacy-preserving digital identity spans multiple domains, from decentralised identity management to advanced cryptographic techniques for credential protection. 
In this section, we frame our contribution within the existing literature by reviewing prior work on Self-Sovereign Identity (\SSI), selective disclosure mechanisms, and the cryptographic primitives underlying modern identity systems. 
For more comprehensive overviews of SSI frameworks and their evolution, we refer the reader to \cite{manimaran2025decentralization}.
We also refer the interested reader to \cite{pava2024self} and \cite{ferdous2019search}, which highlight the interplay between SSI and distributed ledgers, cryptographic protocols, and privacy-enhancing technologies.

\subsubsection*{Selective Disclosure.}
The W3C identifies selective disclosure as a fundamental privacy requirement for VCs, but does not specify a concrete method for its implementation~\cite{vcw3c}. 
As a result, different approaches have been proposed in the literature, each with distinct security and usability trade-offs.
One of the simplest is the atomic credential approach \cite{vcw3cGuidelines,SalveLMR22}, in which each claim is issued as a separate VC.
Hash-based approaches \cite{ietf-oauth-selective-disclosure-jwt-22,SalveLMR22,halpin2020vision} commit to claim values via hash functions, while Merkle-tree based methods \cite{saito2021lightweight,mukta2020blockchain,tariq2022cerberus,tian2023authenticated} exploit hash-tree data structures to enable efficient proofs of inclusion. 
Cryptographic techniques such as attribute-based encryption \cite{hernandez2018protecting} or threshold cryptography \cite{SonninoABMD19} have also been adopted to enforce fine-grained disclosure control.  
Also, signature schemes such as CL signatures \cite{camenisch2004signature} and BBS+ \cite{tessaro2023revisiting} provide built-in selective disclosure through zero-knowledge proofs.
Finally, regarding Zero-knowledge, zkSNARK-based constructions, such as zk-creds~\cite{Rosenberg2023zkcreds}, were introduced to enable flexible anonymous credentials while introducing the notion of publicly verifiable credentials.
However, unlike our proposed solution, all the above approaches still allow the Holder to know exactly which attributes are disclosed to the Verifier: in zero-knowledge-based selective disclosure the predicates on the claims that must be proved by the Holder reveal the internal decision-making logic of the Verifier, and even anonymous credential systems with hidden policies require the Holder to participate in policy-dependent proof generation.
Also, except for the atomic solution (which lacks flexibility), all the cited approaches introduce communication and computational overheads, forcing all the parties to be compliant with changes in the protocols. Among these, hash-based solutions offer the best trade-off between deployment simplicity, computational efficiency, and claim-level granularity, which explains their wide adoption in current SSI frameworks (\eg, SD-JWT~\cite{ietf-oauth-selective-disclosure-jwt-22}).
Indeed, from the Verifier's perspective, a binary reveal/non-reveal model offers significant advantages in terms of simplicity, low computational overhead, and enhanced privacy protection.

\subsubsection*{Key Derivation and Encryption in Identity Systems.}
Beyond credential structures, the secure management of cryptographic keys is central to SSI deployments.  
Key Derivation Functions (KDFs, \cite{Krawczyk2010}), typically based on HMAC or extract-then-expand constructions such as HKDF \cite{rfc5869}, are widely used to derive unlinkable keys for authentication and encryption in protocols such as TLS~1.3~\cite{TLS13} and Noise~\cite{perrin2018noise}, which are increasingly integrated into decentralised identity frameworks~\cite{Perugini2024}.  

Authenticated encryption schemes with associated data (AEAD, \cite{Rogaway2002}), particularly employing Galois Counter Mode (GCM, \cite{mcgrew2004galois}) such as AES-GCM \cite{mcgrew2004galois}, are the \emph{de facto} standard for protecting data storage and transmission, thanks to their efficiency and hardware support~\cite{TLS13}.
Their adoption in identity frameworks ensures both confidentiality of attributes and binding to auxiliary metadata (\eg, credential identifiers or revocation information), which is a fundamental requirement for selective disclosure mechanisms.
Variants, such as ChaCha20-Poly1305~\cite{rfc8439} or Ascon~\cite{Dobraunig2021}, have also been proposed and adopted in practice to provide performance advantages on constrained devices, a relevant setting for mobile identity wallets.

\subsubsection*{Oblivious Key Derivation.}
While traditional KDFs deterministically derive keys from secret material, recent research explores \emph{oblivious} KDFs, in which a client obtains a derived key without revealing its input to the KDF server, and the server learns nothing about the derived key itself~\cite{Casacuberta2022}. These constructions are closely related to Oblivious Pseudorandom Functions (OPRFs, \cite{NaorReingold97}), and provide strong privacy guarantees: each derived key is computationally independent and unlinkable to other derivations, even when the server observes multiple queries.

A widely used real-world OPRF is the discrete-log variant called 2-Hash Diffie Hellman (2HashDH / HashedDH), originally introduced for threshold password authentication key-exchange (T-PAKE~\cite{JareckiKiayiasKrawczyk14}). 2HashDH is simple yet efficient, achieves optimal round complexity, and provides the strongest Universal Composability (UC) security guarantees under the one-more Diffie-Hellman assumption in the random oracle model~\cite{BeullensDodgsonFallerHesse24}, relying on a blind-exponentiate-unblind approach. 
We refer the interested reader to the SoK by Casacuberta et al. \cite{Casacuberta2022} for an in-depth discussion on the properties and constructions of OPRFs.

Although our proof-of-concept adopts 2HashDH due to its simplicity, the construction is independent of the chosen OPRF. In particular, it is worth mentioning the 2HashOPRF framework~\cite{BeullensDodgsonFallerHesse24}, which overcomes the major limitation of 2HashDH by providing post-quantum security.


\section{Preliminaries}
\label{sec:preliminaries}

In this section, we introduce the main building blocks that underlie our construction, namely decentralised identifiers, verifiable credentials, selective disclosure techniques, key derivation and pseudorandom functions, oblivious pseudorandom functions, and authenticated encryption.

\subsubsection*{Self-Sovereign Identity.}
\label{ssec:ssi}

In \SSI, users can create global unique identifiers called \emph{Decentralised IDentifiers} (\DIDs)~\cite{didw3c}, to interact with other entities.  
Each \DID identifies a \DID \emph{subject} and is associated with a \DID \emph{document} describing the subject properties, including the cryptographic material used to prove control over the \DID and to verify assertions.

Attributes of an entity are granted by Issuers via \emph{Verifiable Credentials} (\VCs)~\cite{vcw3c}.
A credential $\VC = (C_\VC, M_\VC, P_\VC)$ consists of a claim--value dictionary $C_\VC = \{c_i : v_i\}$, metadata $M_\VC$ (Issuer \DID, credential type, validity, etc.), and a proof $P_\VC$ obtained by signing $C_\VC \| M_\VC$ with the Issuer key.  
The Holder stores the VCs locally and can create \emph{Verifiable Presentations} (\VPs) by bundling a subset of \VCs with additional metadata and signing them in a triple $\VP = (\VC[\cdot], M_\VP, P_\VP)$.
Verification of a \VC or \VP involves resolving the Issuer and Holder \DIDs to recover the cryptographic material for signature validation.

\subsubsection*{Selective Disclosure.}

A fundamental privacy property in SSI is \emph{selective disclosure}, which allows the Holder to reveal only a subset of claims while maintaining Issuer integrity guarantees.  
Approaches range from atomic credentials and hash-based commitments to Merkle trees, encryption-based solutions, and signature schemes with built-in disclosure support (see Section~\ref{sec:relatedWork}).  

In this work, we focus on the \emph{hash-based approach}, as it provides a simple, commonly-implemented baseline later adaptable to our oblivious variant\footnote{Other selective-disclosure mechanisms, such as the encryption-based construction, particularly fit to our model. However, their comparatively limited adoption in current \SSI frameworks reduces their immediate deployability.}.
Here, claims $v_i$ in $C_\VC$ are replaced by commitments $x_i = H(v_i\|t_i)$, for some hash function $H : \FF_2^\star \to \XXX$ and some salt $t_i$, creating a new claim-digest dictionary $C_\VC^h$; the Issuer then delivers the actual data (couples $v_i, t_i$) as complementary dictionary $D_\VC$ under a (possibly) separate channel.
Digests are typically generated under well-known schemes, but no specific implementation is typically enforced \cite{rfc9901}: our choice falls on SHA-3~\cite{SHA-3} due to its strong security properties, wide adoption, and versatility.
The idealised pipeline based on selective disclosable \VCs is reported later in Figure~\ref{fig:issuer-holder-verifier} (centre, orange), while the structure of hash-based Selective Disclosable \VCs is presented in Figure~\ref{fig:VC-structure}.
\begin{figure}[t]
    \centering
    \begin{tikzpicture}[
        nome/.style={draw, rectangle, anchor=west, minimum height=1.5em,
  minimum width=5.5cm,fill=yellow!30,text width=5.3cm, align=left},
        node distance=0,
        outer sep=0,
        inner sep=2pt
    ]
        \node[nome, align=center] (N1) at (0, 0) {Verifiable Credential {(\VC)}};
        \node[nome] (N2) [below = of N1] {
            \textbf{Payload} $C_\VC^h$\\[.5em]
            $c_1$ : $x_1 = H(v_1 \| t_1)$\\
            $\dots$\\
            $c_n$ : $x_n = H(v_n \| t_n)$
        };
        \node[nome] (N3) [below = of N2] {
            \textbf{Metadata} $M_\VC$
        };
        \node[nome] (N4) [below = of N3] {
            \textbf{Proof} $P_\VC$\\[.5em]
            $\text{sign}(C_\VC^h \| M_\VC, \sk_\text{issuer})$
        };

        \node[nome, align=center] (C1) at (6, 0) {Credential Data {($D_\VC$)}};
        \node[nome] (C2) [below = of C1] {
            \textbf{Payload} $C_\VC$\\[.5em]
            $c_1$ : $v_1, t_1$\\
            $\dots$\\
            $c_n$ : $v_n, t_n$
        };
    \end{tikzpicture}
    \caption{Hash-based selective disclosure structure: claims are committed with nonce $t_i$ inside the VC as values $x_i$ , while clear values $v_i$ are separately delivered as $D_\VC$.}
    \label{fig:VC-structure}
\end{figure}

\subsubsection*{Key Derivation Functions and Pseudorandom Functions.}
\label{ssec:kdf}

Modern protocols often require deriving independent symmetric keys from shared secrets or long-term master keys.  
A \emph{Key Derivation Function} (KDF) transforms keying material into one or more cryptographically strong outputs.  
Typical goals are entropy stretching and domain separation, with HMAC-based KDF (HKDF)~\cite{rfc5869,Krawczyk2010} being a widely deployed construction.

More formally, a \emph{Pseudorandom Function} (PRF) is a keyed family of functions $F: \KKK \times \XXX \to \YYY$ such that, for a random key $k \in \KKK$, $F_k(\cdot)$ is computationally indistinguishable from a random function.  
PRFs abstract the core of KDFs, MACs, and other symmetric-key primitives.

\subsubsection*{Oblivious Pseudorandom Functions.}
\label{ssec:oprf}

An \emph{Oblivious PRF} (OPRF)~\cite{NaorReingold97,Freedman2005} extends the concept of PRF by distributing the evaluation between a client (holding $x \in \XXX$) and a server (holding $k \in \KKK$).  
The client obtains $F_k(x)$ , while the server learns nothing about $x$ and the client learns nothing about $k$.  
To satisfy security objectives, one needs to evaluate correctness along with client and server privacy.

Among the various instantiations, a particularly practical and widely studied candidate is 2HashDH, which we adopt as the core OPRF primitive in our construction.
In 2HashDH, the client blinds the (hash of) its input with a random group element before sending it to the server, which exponentiates under its secret key.  
After unblinding, the client obtains a group element that is independent of the server’s secret, and applies an independent hash function to derive the final pseudorandom output.  
This double hashing step is crucial to prove the UC security of the protocol, preventing linkage across sessions and ensuring pseudorandomness even under adaptive queries. Figure \ref{fig:2HashDH} shows the steps performed by the client and server in the 2HashDH approach.
Namely, given a value $x \in \XXX$, a key $k \in \KKK$, a group $G$, and two hash functions $H_1 : \XXX \to G, H_2: \XXX \times G \to \YYY$:
\begin{enumerate}
    \item The Client randomly chooses $r \in G$ and evaluates and sends $a = H_1(x)^r$;
    \item The server evaluates and sends $b = a^k$;
    \item The Client unblinds $b$ as $c = b^{\sfrac1r}$ and produces the output $y=H_2(x, c)$.
\end{enumerate}
Security relies on the one-more Diffie–Hellman assumption in the Random Oracle Model (ROM), yielding optimal round complexity and strong composable guarantees.  
As we will adopt 2HashDH in our construction, its simplicity and efficiency make it a natural baseline, while post-quantum variants such as 2HashOPRF~\cite{BeullensDodgsonFallerHesse24} demonstrate how the approach can be extended with lattice-based oblivious transfer to achieve long-term security.

\begin{figure}[t]
    \centering

    \input{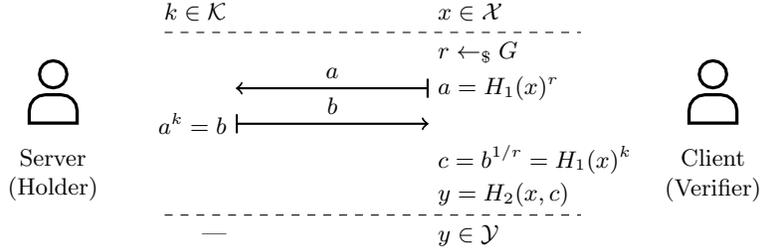}
    
    \caption{The 2HashDH approach for OPRF evaluation.}
    \label{fig:2HashDH}
\end{figure}

\subsubsection*{Authenticated Encryption with AES-GCM.}
\label{ssec:aes-gcm}

In our constructions, keys derived from KDFs or OPRFs are ultimately used within a symmetric encryption scheme that ensures both confidentiality and authenticity.  
The predominant choice in practice is the \emph{Advanced Encryption Standard in Galois/Counter Mode} (AES-GCM)~\cite{mcgrew2004galois}, widely adopted in TLS~1.3 and IPsec for its efficiency, standardisation, and hardware support.

AES-GCM combines AES in counter mode for parallelisable encryption with a polynomial-based authentication function (GHASH) over $\FF_{2^{128}}$.  
Formally, for a secret key $K$, nonce $\IV$, message $M$, and associated data $A$, the encryption is defined as $\Enc_K(\IV,M,A) = (C,T)$, where $C$ is the ciphertext and $T \in \FF_{2^{128}}$ is a 128-bit authentication tag.
Decryption follows naturally as $\Dec_K(\IV, C, T, A) = (M \lor \bot)$, where $T$ is checked against the GHASH of $A$, $C$, $\IV$ under key $K$ prior to the actual ciphertext decryption.

Security of AEAD schemes such as AES-GCM is captured by two complementary notions: \textit{i)} the \emph{IND-CPA confidentiality}, ensuring ciphertext indistinguishability under chosen-plaintext attacks, and \textit{ii)} the \emph{INT-CTXT integrity}, guaranteeing unforgeability of ciphertext--tag pairs.
AES-GCM achieves both properties under standard assumptions, provided nonces are never reused~\cite{mcgrew2004galois}. Its efficiency, compactness, and widespread deployment make it the \textit{de facto} AEAD primitive in practice.


\section{The \CODSSI Model}
\label{sec:COD-SSI}

We are now ready to introduce \CODSSI, our novel approach allowing Verifiers to collect selectively disclosable claims in an oblivious framework.
The methodology builds upon the three classes of primitives introduced in Section~\ref{sec:preliminaries}, enabling a Verifier to selectively extract a given number of claims from a \VP in a privacy-preserving way, while ensuring that all the undisclosed ones remain cryptographically protected and unlinkable.

Typically, selective disclosure is initiated by the Holder as a result of a specific service request. The expected interaction proceeds as follows:
\begin{enumerate}
    \item the Holder selects the set of $N$ claims he is willing to expose within the context of that request;
    \item the Verifier declares that access requires $N_o$ claims from this admissible set, without revealing which ones;
    \item the Holder prepares a \VP in which each claim value is encrypted under a key derived via the OPRF from its hash commitment; and
    \item during verification, the Verifier engages in the OPRF protocol to recover only the $N_o$ keys corresponding to its selected claims, while the Holder learns nothing beyond the number $N_o$ of claims queried.
\end{enumerate}
This design ensures that the Holder retains full control over which claims are eligible for verification, while the Verifier's specific selection remains hidden.
\CODSSI is designed with two complementary objectives.
Given a \VP:  
\begin{itemize}
  \item \textbf{Selective disclosure with obliviousness.}  
  A Verifier can decide which claims (typically $N_o$ out of $N$, $N_o \leq N$) to reveal, while their choice remains unknown to the Holder, who only learns the number of claims queried.  
  \item \textbf{Oblivious key derivation.}  
  Each claim is protected by a key derived through an OPRF, ensuring that even a malicious Verifier cannot learn or correlate non-disclosed claims, and that ciphertexts are indistinguishable from encryptions under fresh random keys.  
\end{itemize}

In this way, our approach advances the state of selective disclosure systems by integrating oblivious primitives into the core cryptographic design.  
Compared to existing hash-based or encryption-based disclosure methods, \CODSSI guarantees that {every ciphertext is independently protected by a hidden, unlinkable key}, while still allowing for efficient verification of the disclosed subset.
It should be noted that Verifier privacy cannot be achieved simply by requesting a superset of claims larger than the $N_o$ actually required: while this may hide the true subset, it forces disclosure of $N^\star \gg N_o$ claims, thereby violating the data-minimisation principles that are central to SSI~\cite{GDPR}.
\CODSSI avoids this trade-off by ensuring that the Verifier accesses exactly the $N_o$ claims it needs, while the Holder cannot determine which $N_o$ out of $N$ were selected.
Moreover, in regulated contexts, the parameters $N$ and $N_o$ can themselves be publicly described and verified by a supervisory authority (\eg, a Data Protection Authority under the GDPR).
For instance, a credit scoring company may keep its scoring function secret while certifying, via regulatory audit, that exactly $N_o$ claims are needed for the evaluation.
\CODSSI enforces this minimisation constraint at the protocol level, without requiring the Verifier to reveal its internal policy, thus bridging the gap between regulatory transparency and business confidentiality.

We next formalise the parties involved in our \CODSSI, we describe its design, and we specify the data structures and algorithms implementing it.

\subsection{Main parties}
\label{ssec:parties}

Our system follows the standard \SSI roles: an \emph{Issuer} who issues Verifiable Credentials (\VCs), a \emph{Holder} who stores \VCs and prepares Verifiable Presentations (\VPs), and a \emph{Verifier} who requests evidence about selected claims.  
Crucially, the oblivious selective-disclosure mechanism introduced here is executed \emph{exclusively} between the Holder and the Verifier at presentation time (\ie the Issuer is not informed, nor does it participate in the OPRF-based key derivation or in the selective reveal protocol).
Consequently, the Issuer's role is limited to standard credential issuance and signature binding to the credential commitments $C_\VC^h$; all subsequent oblivious interactions are Holder--Verifier only.  
This design choice preserves the Issuer's simplicity and avoids requiring any online involvement of credential issuers at verification time.

In detail, the parties involved in \CODSSI are:
\begin{itemize}
  \item the Issuer $\IIII$, who issues a credential $\VC=(C_\VC^h,M_\VC,P_\VC)$ with corresponding data $D_\VC$ to a Holder $\HHHH$ in an offline issuance phase;
  \item the Holder $\HHHH$, who stores the credential(s), prepares a Verifiable Presentation \VP with corresponding presentation data $D_\VP$ containing encrypted claims $E_\VC[\cdot]$ and who acts as the OPRF \emph{server} during verification;
  \item the Verifier $\VVVV$, who inspects the presentation metadata and acts as the OPRF \emph{client} to obtain the claim-specific keys for the claims she wishes to learn.
\end{itemize}
We assume that the Holder generates and maintains an OPRF secret key $\msk_\VP \in \KKK$ (or otherwise controls the server-side secret material for the OPRF instance associated with the presentation), and that the Verifier has access to the Holder to run the OPRF protocol.  
The encrypted claim blobs $E_\VC$ are transported as part of $(\VP,D_\VP)$ prior to any OPRF interaction.

\subsection{Creation of the Verifiable Presentation}
\label{ssec:VP-creation}

The core idea of \CODSSI is to create a set of encrypted claims $E_\VC$ that can be safely shared with a Verifier during \VP presentation.
Upon generating a \VP master key $\msk_\VP$, the Holder can use a KDF to derive, for each claim $v_i$, a claim key $k_i$ from some claim-related identifier, like the hash $x_i = H(v_i \| t_i)$ in $C_\VC^h$, which is already present in the \VC itself.
This way, confidentiality reduces to the problem of obliviously sharing claim keys, ensuring that undisclosed claims remain unlinkable and indistinguishable from random encryptions.

More formally, the Holder wishing to prepare a presentation collects the selective disclosable \VCs in a \VP,  as usual.
However, conversely to the selective disclosable case, she generates a master secret key $\msk_\VP \in \KKK$ for the OPRF and, for each credential $\VC = (C_\VC^h, M_\VC, P_\VC, D_\VC)$, proceeds as follows:
\begin{enumerate}
    \item \textbf{Key derivation.}
    For each claim $\{c_i : v_i\}$, the Holder retrieves $x_i = H(v_i \| t_i)$ from $C_\VC^h$.
    The corresponding claim key is derived as
    \[
        k_i = F_{\msk_\VP}(x_i)
         \,,
    \]
    where $F$ denotes the OPRF. Hence, under 2HashDH (cf.\ Section~\ref{ssec:oprf}), we have
    \[
        k_i = H_2(x_i, H_1(x_i)^{\msk_\VP}) \,.
    \]

    \item \textbf{Encryption.}  
    A fresh nonce $\IV_i$ is sampled, and the randomised claim $m_i = v_i \| t_i$ is encrypted as
    \[
        (y_i, u_i) = \Enc_{k_i}(\IV_i, m_i, x_i) \,,
    \]
    where $\Enc$ is an AEAD scheme (AES-GCM in our proof of concept, cf.\ Section~\ref{ssec:aes-gcm}), and $x_i$ is included as associated data.
    
    \item \textbf{Packaging.}  
    The Holder collects the triples $(\IV_i, y_i, u_i)$ in the encrypted set
    \[
        E_\VC = \{\, (x_i : (\IV_i, y_i, u_i)) \,\}_i \,,
    \]
    which is then embedded in the presentation data $D_\VP$.
\end{enumerate}
The final verifiable presentation thus consists of the usual selective-disclosure structure, augmented with the encrypted presentation data $E_\VC[\cdot]$ (see Figure~\ref{fig:VP-structure}).
The process is also summarised in Algorithm~\ref{alg:create-vp}.

\begin{figure}[t]
    \centering
    \begin{tikzpicture}[
        nome/.style={draw, rectangle, anchor=west, minimum height=1.5em,
  minimum width=5.5cm,fill=yellow!30,text width=5.3cm, align=left},
        node distance=0,
        outer sep=0,
        inner sep=2pt
    ]
        \node[nome, align=center] (N1) at (0, 0) {Verifiable Presentation {(\VP)}};
        \node[nome] (N2) [below = of N1] {
            \textbf{Payload} $\VC[\cdot]$\\[.5em]
            1 : $\VC_1 = (C^h_{\VC_1}, M_{\VC_1}, P_{\VC_1})$\\
            $\dots$\\
            $m$ : $\VC_m = (C^h_{\VC_m}, M_{\VC_m}, P_{\VC_m})$
        };
        \node[nome] (N3) [below = of N2] {
            \textbf{Metadata} $M_\VP$
        };
        \node[nome] (N4) [below = of N3] {
            \textbf{Proof} $P_\VP$\\[.5em]
            $\text{sign}(\VC[\cdot] \| M_\VP, \sk_\text{holder})$
        };

        \node[nome, align=center] (A1) at (6, 0) {Presentation Data {($D_\VP$)}};
        \node[nome] (A2) [below = of A1] {
            \textbf{Payload} $E_\VC[\cdot]$\\[.5em]
            $1$ : $E_{\VC_1}$\\
            $\dots$\\
            $m$ : $E_{\VC_m}$
        };

        \node[nome, align=center] (C1) at (6, -2.5) {Encrypted Credential Data {($E_\VC$)}};
        \node[nome] (C2) [below = of C1] {
            $c_1$ : $(\IV_1, \Enc_{k_1}(\IV_1, v_1 \| t_1, x_1))$\\
            $\dots$\\
            $c_n$ : $(\IV_n, \Enc_{k_n}(\IV_n, v_n \| t_n, x_n))$
        };
    \end{tikzpicture}
    \caption{Structure of the standard selective disclosable Verifiable presentation (\VP) on the left and of our encrypted presentation data on the top right.}
    \label{fig:VP-structure}
\end{figure}

\begin{algorithm}[t]
    \caption{Creation of a \VP}
    \label{alg:create-vp}
    \KwInput{A set of selective-disclosable credentials $\{\VC_j =$ $(C^h_{\VC_j}, M_{\VC_j}, P_{\VC_j}, D_{\VC_j})\}_{j=0}^{m}$ to include in the \VP.}
    \KwOutput{A Verifiable Presentation $\VP = (\VC[\cdot], M_\VP, P_\VP)$ and encrypted material $D_\VP = (E_\VC[\cdot])$.}
    \BlankLine

    \tcp{1. Initialise the presentation}
    Generate a fresh master secret key $\msk_\VP \leftarrow_\$ \KKK$\;
    Initialise $D_\VP \gets \emptyset$\;

    \BlankLine
    \tcp{2. Process credential}
    \ForEach{$j=0, \dots, m$}{
        
        Initialise $E_{\VC_j} \gets \emptyset$\;

        \BlankLine
        \ForEach{$i = 0, \dots, n_j$}{
            \tcp{2a. Key derivation for each claim}
            Retrieve $c_i, v_i, t_i$ from $D_{\VC_j}$\;
            Retrieve $x_i = C^h_{\VC_j}[c_i]$\;
            Compute claim key $k_i = H_2(x_i, H_1(x_i)^{\msk_\VP})$\;

            \BlankLine
            \tcp{2b. Encrypt the randomised claim}
            Let $m_i = v_i \| t_i$\;
            Sample fresh nonce $\IV_i$\;
            Compute $(y_i, u_i) = \Enc_{k_i}(\IV_i, m_i, x_i)$\;

            \BlankLine
            \tcp{2c. Package encrypted material}
            Add entry $(c_i : (\IV_i, y_i, u_i))$ to $E_{\VC_j}$\;
        }

        \BlankLine
        Add $E_\VC$ to $D_\VP$\;
    }

    \BlankLine
    \tcp{3. Finalise VP metadata and proofs}
    Construct presentation metadata $M_\VP$\;
    Sign presentation proof $P_\VP$\;

    \BlankLine
    \Return $\VP = (\VC[\cdot], M_\VP, P_\VP)$ and $D_\VP = (E_\VC[\cdot])$\;

\end{algorithm}

\subsection{Verification of the Verifiable Presentation}
\label{ssec:VP-verification}

Upon receiving a verifiable presentation $(\VP, D_\VP)$, the Verifier validates the \VP and the corresponding \VCs proofs by checking that the \VCs or \VP conforms to the specification, that the Holder's signature on the VP is valid, that the Holder is the subject of the \VCs, and whether they are expired or revoked. The Verifier then inspects the metadata, and checks the paths of the included \VCs, as in the selective disclosure framework.
In contrast to the standard case, the Verifier does not request plaintext claims directly.  
Instead, for each claim of interest, she engages in the OPRF evaluation with the Holder, thereby obtaining the corresponding decryption keys $k_i$ without revealing which claims were selected.  
This guarantees that the Verifier's choice remains oblivious, while the Holder retains the ability to enforce disclosure policies, such as limiting the number of claims that can be revealed.

Depending on the Holder's disclosure policy, the OPRF evaluation may be performed in different modes, like, \eg:
in a \emph{Batch mode}, where the Verifier requests all desired claim keys in a single interaction, minimising round complexity; or in \emph{Adaptive mode}, where the Verifier engages in multiple rounds, choosing further claims sequentially based on the results of previously disclosed ones.
Both modes preserve obliviousness, but differ in flexibility and control: adaptive disclosure allows incremental negotiation, while batch disclosure improves efficiency.  
Either way, the Holder may enforce policies such as halting the interaction after a disclosure quota is reached (\emph{quota exhaustion}).

The overall flow of \CODSSI, w.r.t.\ the classical and the selective disclosable framework, is summarised in Figure~\ref{fig:issuer-holder-verifier} (bottom, blue). The algorithmic solutions for adaptive and batch modes are presented in Algorithms~\ref{alg:verify-vp-oprf} and~\ref{alg:verify-vp-oprf-batch} respectively.

\begin{algorithm}[t]
    \caption{Verification of a \VP (Adaptive)}
    \label{alg:verify-vp-oprf}
    \KwInput{Verifiable Presentation $\VP = (\VC[\cdot], M_\VP, P_\VP)$, encrypted material $D_\VP = (E_\VC[\cdot])$, and quota $N_o$.}
    \KwOutput{Decrypted claims of interest.}
    \BlankLine
    
    \tcp{1. Check consistency}
    Verifier validates $P_\VP$\;
    \ForEach{$\VC = (C^h_\VC, M_\VC, P_\VC)$ in $\VC[\cdot]$}{
        Verifier validates $P_\VC$\;
    }

    \BlankLine
    \tcp{2. Execute adaptive interaction with Holder}
    \ForEach{$j=0, \dots, N_o-1$}{
        Select $VC_i$ from \VP\;
        Select a claim $c$ in $\VC_i$\;
        Retrieve claim hash $x = C^h_{\VC_i}[c]$\;
        Generate $r \leftarrow_\$ G$\;
        Send $a = H_1(x)^r$ to Holder\;
        Receive $b = a^{\msk_\VP}$ from Holder\;
        Evaluate the claim key $k = H_2(x, b^{1/r})$\;
        Retrieve $(\IV, y, u) = E_{\VC_i}[c]$\;
        Check correctness of tag $u$\;
        Decrypt $m_j$ from $y$ under key $k$\;
        Verify $H(m_j) = x$\;
    }

    \BlankLine
    \Return $\{m_{j}\}_{j=0}^{N_o}$\;
    
\end{algorithm}

\begin{algorithm}[t]
    \caption{Verification of a \VP (Batch)}
    \label{alg:verify-vp-oprf-batch}
    \KwInput{Verifiable Presentation $\VP = (\VC[\cdot], M_\VP, P_\VP)$, encrypted material $D_\VP = (E_\VC[\cdot])$, and quota $N_o$.}
    \KwOutput{Decrypted claims of interest.}
    \BlankLine

    \tcp{1. Check consistency}
    Verifier validates $P_\VP$\;
    \ForEach{$\VC = (C^h_\VC, M_\VC, P_\VC)$ in $\VC[\cdot]$}{
        Verifier validates $P_\VC$\;
    }

    \BlankLine
    \tcp{2. Prepare Batch request}
    Select a set of claims $S = \{(c_j, i_j)| c_j \in C^h_{VC_{i_j}}\}_{j=0}^{N_o}$\; 
    \ForEach{$(c, i) \in S$}{
        Retrieve claim hash $x_j = C^h_{\VC_i}[c]$\;
        Sample fresh blinding exponent $r_j \leftarrow_\$ G$\;
        Compute $a_j = H_1(x_j)^{r_j}$\;
    }

    \BlankLine
    \tcp{3. Execute batch interaction with Holder}
    Send $\{a_j\}_{j=0}^{N_o}$ to Holder\;
    Receive $\{b_j\}_{j=0}^{N_o}$ from Holder, where $b_{j} = a_{j}^{\msk_\VP}$\;

    \BlankLine
    \tcp{4. Derive key \& decrypt}
    \ForEach{$j=0, \dots, N_o-1$}{
        Compute claim key $k_j = H_2(x_j, b_j^{1/r_j})$\;
        Retrieve $(\IV_j, y_j, u_j) = E_{\VC_i}[c_j]$, $i=i_j$\;
        Check correctness of tag $u_j$\;
        Decrypt $m_j$ from $y_j$ under key $k_j$\;
        Verify $H(m_j) = x_j$\;
    }

    \BlankLine
    \Return $\{m_{j}\}_{j=0}^{N_o}$\;

\end{algorithm}

\subsection{Asymptotic Complexity}
\label{ssec:complexity}

Table~\ref{tab:complexity} summarises the asymptotic computational and communication costs of \CODSSI.
We recall that $N$ represents the total number of claims in the \VP, $N_o \leq N$ the number of claims disclosed during a session, and $m$ the number of \VCs bundled in the \VP.
The dominant cost in all phases is given by the scalar multiplication in the elliptic-curve group $G$ (denoted $\mathsf{Exp}$).
Hash-to-group evaluations ($H_1$), standard hashes ($H$, $H_2$), and symmetric cryptographic operations (AEAD encryption/decryption) are significantly cheaper and are reported separately.

\begin{table}[t]
\centering
\caption{Asymptotic costs of \CODSSI per phase and party.}
\label{tab:complexity}
\setlength{\tabcolsep}{5.2pt}
\renewcommand{\arraystretch}{1.3}
\begin{tabular}{llccc}
    \toprule
    \textbf{Phase} & \textbf{Party} & \textbf{Group ops} & \textbf{Sym.\ ops} & \textbf{Comm.}\\
    \midrule
    VP creation     & Holder   & $N\,(\mathsf{Exp} + H_\textsc{g})$  & $N (\mathsf{Enc} + H_2)$  & ---\\
    VP validation   & Verifier & $(m{+}1)\,\mathsf{SigV}$      & ---                          & ---\\
    \midrule
    OPRF discl. & Verifier & $N_o\,(2\,\mathsf{Exp} + H_\textsc{g})$ & $N_o (\mathsf{Dec} + H + H_2)$ & $N_o|G|$\\
                    & Holder   & $N_o\,\mathsf{Exp}$             & ---                          & $N_o|G|$\\
    \midrule
    Storage ($D_\VP$) & ---    & ---                   & ---                          & $N\!\cdot\!(|\IV|{+}|\textsf{ct}|{+}|\textsf{tag}|)$\\
    \bottomrule
\end{tabular}

\medskip

$\mathsf{Exp}$: scalar multiplication in $G$; $H_1$: hash-to-group; $H$, $H_2$: standard hashes;\\ $\mathsf{Enc}, \mathsf{Dec}$: AEAD operations; $\mathsf{SigV}$: signature verification.
\end{table}

Several observations follow.
First, \VP creation scales linearly in $N$: for each claim the Holder performs one hash-to-group evaluation $H_1(x_i)$, one scalar multiplication $H_1(x_i)^{\msk_\VP}$, one hash $H_2$ for key finalisation, and one AEAD encryption.
Second, the disclosure cost depends on $N_o$, not $N$ (but for the signature verification): since $N_o \ll N$ in typical scenarios, the Verifier and Holder engage in a small number of OPRF interactions regardless of how many claims the \VP contains.
Third, the Verifier's per-claim cost ($2\,\mathsf{Exp}$, for blinding and unblinding) is roughly twice the Holder's ($1\,\mathsf{Exp}$, for the server-side exponentiation).
Finally, batch and adaptive modes yield the same total computation and communication; they differ only in round complexity ($1$ round vs.\ $N_o$ rounds).

\section{\CODSSI security}
\label{sec:codssi-security}

This section analyses the security of our construction under the standard compositional methodology: since the employed primitives are proven secure in well-established frameworks, the overall scheme inherits these properties under black-box composition.

We first formalise the adversarial setting and the threat model, then present the security goals that our verifiable presentation protocol with oblivious selective disclosure aims to achieve.
We subsequently state and proof-sketch the main security theorem, derived using standard assumptions on the underlying primitives (a formal analysis is provided in Appendix~\ref{app:security}).
Finally, we discuss the limitations of the current construction when internal parties become malicious, with particular attention to the vulnerability arising from a malicious Holder; we explain its implications and outline possible future directions to overcome it.

\subsection{Threat Model and Adversarial Setting}

As previously described in Section~\ref{ssec:parties}, our protocol involves three internal roles: the \emph{Issuer}, the \emph{Holder}, and the \emph{Verifier}.
Unless stated otherwise, all these parties are assumed to be \emph{honest-but-curious} (HBC): they follow the protocol specification exactly, but may attempt to learn additional information from their allocated views.
This baseline model captures the privacy guarantees that the system provides under compliant participation, which is appropriate for a first-step analysis of selective disclosure mechanisms with oblivious key retrieval.

In addition to internal HBC agents, we assume a malicious \emph{external} adversary $\mathcal{A}$ controlling the communication network.
The adversary can intercept, delay, drop, reorder, and inject messages arbitrarily, and may attempt to compromise confidentiality or integrity via any feasible network-level attack.
Cryptographic security must therefore be achieved despite complete network control: confidentiality and integrity of encrypted claims rely on the AEAD guarantees, and the privacy and correctness of key derivation rely on the OPRF security.
We do not assume confidentiality or integrity of the transport channel.
However, we assume that parties can authenticate each other and bind protocol messages to the relevant session and presentation instance (\eg relying on the on-chain public keys provided by their \DIDs), so that $\mathcal{A}$ cannot successfully impersonate a Holder or Verifier or mount replay attacks beyond trivial denial-of-service (see the operational assumptions below for details).

This threat model allows us to isolate the intrinsic guarantees of the protocol design.
The case where internal parties are malicious is considered separately in Section~\ref{ssec:malicious-internal}, where we discuss the main limitations that emerge once the HBC assumption is relaxed.

\paragraph{Operational assumptions.}
In addition to the cryptographic assumptions stated in Theorem~\ref{thm:security}, we rely on the following operational assumptions that are necessary for the protocol to achieve its security goals in practice:
\begin{enumerate}
    \item \textbf{Fresh presentation key.}
    The Holder generates a fresh OPRF master key $\msk_\VP$ for each new \VP instance.
    This ensures cross-session unlinkability: since the key material is independent across presentations, a Verifier cannot correlate OPRF outputs across distinct sessions, even when querying the same claim index.
    The lifecycle of $\msk_\VP$ is thus bounded to a single \VP: it is generated during \VP creation (Algorithm~\ref{alg:create-vp}), used during the subsequent disclosure interaction, and discarded thereafter.

    \item \textbf{Session binding and replay protection.}
    Each OPRF interaction is bound to a unique session identifier derived from the \VP instance (\eg, combining the \VP nonce, the Verifier's \DID, and a freshly exchanged nonce).
    This prevents replay attacks, where an adversary re-submits a recorded OPRF response to a different session, as well as cross-session splicing, where messages from one protocol run are injected into another.
    Concretely, standard mutual authentication via \DID-bound public keys, combined with a session-specific challenge--response (\eg, as in TLS~1.3~\cite{TLS13} or DID-based protocols~\cite{Perugini2024}), is sufficient to realise this binding.

    \item \textbf{Holder-initiated interaction.}
    We assume that the disclosure protocol is always initiated by the Holder as a consequence of a service request (cf.\ the workflow in Section~\ref{sec:COD-SSI}).
    This means that the Holder never responds to unsolicited verification requests.
    As a consequence, the plausible concern of parallel-session attacks in which multiple Verifiers simultaneously extract claims is mitigated at the application level: the Holder controls the scheduling of sessions and can enforce a quota on the total number of OPRF evaluations across concurrent sessions, \eg, via a shared counter protected by standard concurrency-control mechanisms.

    \item \textbf{Nonce uniqueness.}
    The AEAD scheme (AES-GCM) requires that no nonce $\IV_i$ is ever reused under the same key $k_i$.
    In our construction, nonces are sampled uniformly at random per claim during \VP creation (cf.\ Algorithm~\ref{alg:create-vp}); since $\msk_\VP$ is fresh per presentation, the risk of nonce collision is bounded by the birthday bound on the nonce space and is negligible in practice.
\end{enumerate}

\noindent
On top of such operational assumptions, it is also worth noting two practical leakage vectors that are endemic to our construction and that we decouple from the cryptographic guarantees:
\begin{itemize}
    \item \emph{Policy leakage.} Any Holder-enforced policy that produces observable signals (\eg, quota exhaustion, throttling, or session termination) inevitably leaks coarse information about the disclosure state. Our obliviousness claim (see below Definition~\ref{def:security}, property~3b) is therefore stated \emph{modulo} such policy-observable outputs, which are considered legitimate and known \textit{a priori}.
    \item \emph{Holder compromise.} If an adversary extracts $\msk_\VP$, confidentiality is lost for all claims in the affected presentation. This is an explicit failure mode of the threat model. The fresh-key assumption (item~1 above) bounds the damage to a single \VP instance; further mitigation can be achieved through standard operational practices such as hardware security modules or short-lived key material.
\end{itemize}

\subsection{Security statements}

We now formalise the security guarantees achieved by our construction and we provide a proof sketch.

\begin{definition}[Security of Verifiable Presentations with Oblivious Selective Disclosure]
\label{def:security}
Let $(\VP, D_\VP)$ be a verifiable presentation generated as described in Algorithm~\ref{alg:create-vp}, containing $N$ encrypted claims $\{(x_i, \IV_i, y_i, u_i)\}_{i=1}^{N}$.
Our protocol for \VP verification with oblivious selective disclosure is said to be \emph{secure} if it achieves the following properties against Probabilistic Polynomial Time (PPT) adversaries:
\begin{enumerate}
    \item \textbf{Confidentiality:} for any PPT adversary $\AAA_\VVVV$ that does not obtain the claim key $k_j$ --for some index $j$-- the ciphertext triple $(\IV_j, y_j, u_j)$ is computationally indistinguishable from an AEAD encryption of any equal-length message under a fresh uniform key.
    (Reduces to AEAD IND-CPA security and OPRF pseudorandomness; see also property~3a below.)

    \item \textbf{Authenticity and integrity:} any modification of $(\IV_i, y_i, u_i)$ performed by a PPT adversary is detected by the Verifier with overwhelming probability.
    (Reduces to AEAD INT-CTXT security.)

    \item \textbf{Obliviousness of disclosure (\emph{dual guarantee}):}
    \begin{enumerate}
        \item[\emph{(a)}] \emph{Client privacy:} a PPT Verifier (OPRF client) learns $F_{\msk_\VP}(x_i)$ only for the inputs $x_i$ on which she explicitly queries the OPRF; for all other indices, the OPRF outputs are computationally indistinguishable from independent uniform values.
        (Reduces to OPRF pseudorandomness.)
        \item[\emph{(b)}] \emph{Server privacy:} a PPT Holder (OPRF server) gains no information on which subset of claim keys has been retrieved, beyond what is implied by policy-observable events, like, \eg, quota exhaustion.
        (Reduces to OPRF client-privacy.)
    \end{enumerate}
    \item \textbf{Issuer verifiability:} the validity of each claim remains cryptographically verifiable against the Issuer's signature in the corresponding $\VC$, as in the selective disclosure framework.
\end{enumerate}
\end{definition}

\begin{theorem}[Security of \CODSSI]
\label{thm:security}
Let us assume that:
\begin{enumerate}
    \item the underlying OPRF protocol (2HashDH) is UC-secure against malicious adversaries in the random oracle model (ROM);
    \item the symmetric encryption protocol \Enc (AES-GCM) provides IND-CPA confidentiality and INT-CTXT authenticity under standard assumptions; and,
    \item the selective disclosure mechanism for \VCs (SHA-3) is secure in the ROM.
\end{enumerate}
Then, under the operational assumptions stated in Section~\ref{sec:codssi-security}, the proposed verifiable presentation with oblivious selective disclosure, as described in Sections~\ref{ssec:VP-creation} and \ref{ssec:VP-verification}, satisfies Definition~\ref{def:security} against any PPT malicious external adversary and compliant (HBC) internal parties.
\end{theorem}

\begin{proof}[Proof sketch, see also Appendix~\ref{app:security}]
Let us first recall that the 2HashDH is UC-secure in the ROM~\cite{BeullensDodgsonFallerHesse24}, AES-GCM provides IND-CPA and INT-CTXT security~\cite{mcgrew2004galois,Iwata2012}, and the selective disclosure based on SHA-3 is secure in the ROM~\cite{SHA-3}.
We now proceed deriving the security properties required in Definition~\ref{def:security}.

\emph{Obliviousness of disclosure (property 3).}
The dual guarantee follows from the UC-security of the 2HashDH OPRF.
For part~(a), the OPRF pseudorandomness ensures that the Verifier learns $F_{\msk_\VP}(x_i)$ only for the inputs $x_i$ on which she explicitly queries the OPRF; all other outputs are computationally indistinguishable from independent uniform values.
For part~(b), the OPRF client-privacy guarantee ensures that the Holder, acting as OPRF server, learns nothing about which inputs the Verifier queried.

\emph{Confidentiality and Integrity (properties 1--2).}
By property~3a (client privacy), for every unqueried index $j$ the OPRF output $k_j = F_{\msk_\VP}(x_j)$ is computationally indistinguishable from a uniform random key.
Given this, ciphertexts $y_j$ are indistinguishable from random by the IND-CPA security of AES-GCM (property~1), and any forgery of $(y_j, u_j)$ is detected with overwhelming probability by INT-CTXT tag verification (property~2), even in the presence of an active network adversary.

\emph{Issuer verifiability (property 4)} is directly inherited from the selective disclosure framework: Issuer's signatures remain bound to the claim commitments in $C^h_\VC$, and thus to the encrypted claims $E_\VC$. Any modification to the credential will invalidate the issuer's signature. 

Finally, since AEAD providing both IND-CPA and INT-CTXT properties realises the authenticated encryption functionality under UC composition~\cite{Canetti2022}, our use of AES-GCM inherits UC-style composability.  
Hence, by composition, the overall protocol satisfies Definition~\ref{def:security}.
\end{proof}

\begin{remark}\label{rem:dual-guarantee}
The dual guarantee (properties~3a and~3b of Definition~\ref{def:security}) jointly captures the OT-like privacy structure of our oblivious disclosure protocol.
A fully construction-independent formalisation of these properties as \emph{Verifier-Privacy-Preserving Selective Disclosure} (VPP-SD) via indistinguishability games over an abstract VP scheme syntax is a natural direction for future work, and would allow the analysis of alternative instantiations beyond the OPRF-based design presented here.
\end{remark}

\subsection{Extending the Security to Malicious Actors}
\label{ssec:malicious-internal}

We now discuss the limitations of the construction when internal protocol participants deviate from the honest-but-curious assumption. 
As shown in the previous section, security against an active external adversary holds under standard assumptions on the underlying AEAD, OPRF, and selective disclosure mechanisms. 
We therefore focus on internal parties and show that:
(i) a malicious Issuer has no meaningful impact during the verification phase;
(ii) a malicious Verifier does not gain additional advantage beyond what is already captured by the security definitions; and
(iii) the only non-trivial residual vulnerability arises when the Holder behaves maliciously during \VP creation.

\subsubsection{On the security under malicious Issuer}

The Issuer is not an active participant in the \VP verification protocol, and therefore its behaviour can be modelled in the same way as an external party.
A malicious Issuer cannot interfere with the verification transcript, inject messages, or adaptively influence the selective disclosure process.
Its only possible impact lies in issuing malformed or adversarially generated credentials during \VC creation.

In line with standard selective disclosure frameworks, we assume that Holders (and subsequently verifiers) perform the usual validation steps before relying on any credential:
signature verification on the \VC, consistency checks on the commitments $C_\VC^h$, and structural validation of the credential contents.
Provided these checks succeed, the Issuer has no further role, and its malicious behaviour has no impact on the security of the subsequent oblivious verification phase.

\subsubsection{On the security under malicious Verifier}
A malicious Verifier can deviate arbitrarily during the verification phase, yet this poses no threat beyond what is already captured by the confidentiality and obliviousness guarantees of Definition~\ref{def:security}.
Provided the OPRF construction is secure against malicious clients and the AEAD scheme guarantees IND-CPA and INT-CTXT, the Verifier cannot learn any information about a claim for which she has not obtained the corresponding OPRF-derived key $k_j$.

Formally, a malicious Verifier that fails to derive $k_j$ for some index $j$ has only negligible advantage in distinguishing $(\IV_j, y_j, u_j)$ from an encryption of any other message of equal length, by reduction to OPRF and AEAD security.  
We emphasise however two practical caveats:
(i) OPRF determinism can create linkability across \VPs unless $\msk_\VP$ is chosen fresh (see Section~\ref{sec:codssi-security} regarding key-management assumptions); and
(ii) Holder policy signals and implementation metadata (sizes, timing, error messages) can leak information and must be mitigated (uniform error handling, padding, rate limits).

\subsubsection{On the security under malicious Holder}

A Holder may deviate from the prescribed protocol in two distinct phases: the creation of the \VP and its subsequent verification.
We analyse both cases separately, as the first is the only one posing an actual threat.

\paragraph{\VP verification.}
During the verification protocol, the Holder acts exclusively as the OPRF server (cf.\ line~10 of Algorithm~\ref{alg:verify-vp-oprf} and line~10 of Algorithm~\ref{alg:verify-vp-oprf-batch}).  
Since the underlying 2HashDH OPRF is assumed secure against malicious Servers in the UC framework, a dishonest Holder cannot bias, misdirect, or otherwise influence the values derived by the Verifier beyond what is permitted by the ideal OPRF functionality.
In particular, any deviation by a malicious Holder produces either (i) an output distribution that is indistinguishable from one generated using the ideal OPRF functionality, or (ii) an abort event that the Verifier detects.
Moreover, from the Holder’s perspective, the Verifier’s OPRF input is pseudorandom; hence, no information is leaked through selective failures or adaptive invalid responses.
Consequently, malicious behaviour during verification does not invalidate Definition~\ref{def:security}: confidentiality, authenticity, and obliviousness continue to hold even against a Holder arbitrarily deviating during verification.


\paragraph{VP creation.}
The critical limitation instead concerns the creation of the \VP.
The security theorem requires that $(\VP, D_\VP)$ is generated according to Algorithm~\ref{alg:create-vp}.
This assumption is essential and does not hold under a malicious Holder model: if the Holder manipulates the construction of ciphertexts or the OPRF inputs used for key derivation, the resulting \VP may fail to correspond to the Issuer-certified claims contained in the \VC, actually enabling \textit{selective-failure attacks}, see below.


\subsection{Limitations and Mitigations Under Malicious Holder}
\label{ssec:limitations-mitigations}

As specified in the previous section, the security --and in particular the obliviousness property-- does not hold if the Holder is allowed to tamper with the creation of the \VP.  
More precisely, the current construction does not cryptographically bind together the tuple
\[
    \big(\;v_i\,,\; x_i\,,\; k_i\,,\; (\IV_i, y_i, u_i)\;\big)\,;
\]
therefore, a malicious Holder might, \eg,
(i) encrypt a modified claim value $v'_i \neq v_i$ while retaining the Issuer commitment $x_i$ unchanged or
(ii) derive the encryption key $k_i$ from a manipulated OPRF input $x_i'$ instead of the prescribed $x_i$.
Both these deviations produce AEAD-valid ciphertexts that break the semantic consistency between the \VC and the resulting \VP.
This violates the assumption that the \VP was created honestly and opens space to a potential Verifier privacy breach:
if (i) the specific claim is chosen for decryption, the AEAD check fails and the Verifier will abort the process, revealing its choice; else 
if (ii) the claim was not chosen, the Verifier does not notice the misconduct and the Holder learns the claim was not chosen.
In summary, the Verifier cannot enforce consistency among plaintexts, key material, and ciphertexts at presentation time, and therefore cannot detect dishonest \VP creation prior to key exchange.


\subsubsection{Possible mitigations.}
Several approaches could enforce correct \VP construction, though each introduces additional assumptions or computational overhead and lies outside the current proof-of-concept scope:

\begin{enumerate}
    \item \emph{Issuer-assisted VP generation.}  
    The Issuer participates in or validates the construction of the encrypted claims, e.g.\ by issuing signatures over commitments that bind $(v_i, x_i, \IV_i)$ and approving the OPRF inputs used for key derivation.  
    This restores soundness but requires Issuer involvement at presentation time and may force the Issuer to handle OPRF-related material, which becomes particularly cumbersome in multi-issuer \VPs.

    \item \emph{Trusted VP-generation environment.}  
    A trusted wallet engine or secure enclave enforces correct linkage between Issuer commitments, plaintext claims, key derivation, and AEAD ciphertexts.  
    This avoids Issuer involvement but introduces a trusted-software assumption that only partially aligns with SSI principles.

    \item \emph{Zero-knowledge proofs of correct encryption.}  
    The Holder produces a proof that each ciphertext correctly encrypts the Issuer-committed claim using the key derived from the prescribed OPRF input.
    This achieves full cryptographic binding and soundness without additional trust assumptions, but at significant computational cost.
\end{enumerate}

\section{Experimental Results}
\label{sec:exp}
To assess the practical performance of our framework, we conducted an experimental campaign.  
All benchmarks were executed on a Linux personal laptop equipped with an Intel Core i7-9750H CPU @ 2.60GHz, 16GB RAM. The PoC implementation was developed in JavaScript, ensuring reproducibility in a stable environment, and experiments were run with an idle machine.

Each experiment considered a \VC with $n \in \{2,4,8,\ldots,1024\}$ claims, where each claim was represented by a random hex string of 30 bytes.
For every configuration, we performed $1000$ repetitions of all the protocols, discarding the top and bottom $1\%$ outliers to mitigate fluctuations and obtain statistically meaningful averages.
The evaluation covered three main aspects:
\begin{enumerate}
\item \textbf{Claim level:} hashing, verification, encryption, decryption, and OPRF interaction time per claim.
\item \textbf{Credential level:} creation and verification times, together with the size of \VCs and their metadata.
\item \textbf{Presentation level:} creation and verification time and size of \VPs, the additional overhead of $D_\VP$, and the cost of the oblivious interaction between Verifier and Holder.
\end{enumerate}

\noindent\textbf{Implementation details.} 
The OPRF group $G$ was set to the elliptic curve \texttt{secp256k1} and we used SHA3-512 for implementing $H, H_1,$ and $ H_2$, with suitable group mapping as provided by \texttt{noble} library \cite{noble}.
AES-GCM was used as the AEAD encryption scheme, with random nonce of 1024 bits generated per claim, 256 bits key length, and a 96 bits initialisation vector (\IV).
For decentralised identity management, the implementation adopts the \texttt{did:ethr} \cite{etherDID} method based on a local Ethereum blockchain, while W3C verifiable credentials and presentations are represented in JWT format \cite{jwtvc}. Selective disclosure is built under SHA-3 according to the SD-JWT format~\cite{rfc9901}.

\begin{figure*}[tb]
    \centering
    \includegraphics[width=\dxplot\linewidth]{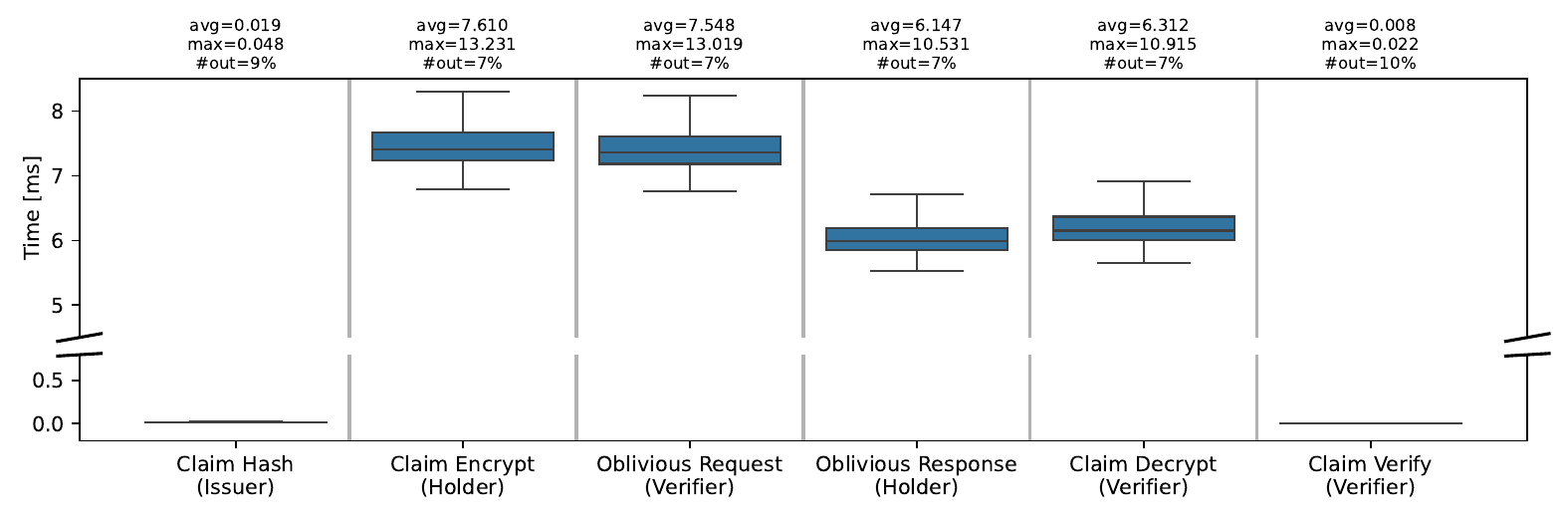}
    \caption{Claim-wise statistics of the computational time for: claim hashing, verifying, encrypt, and decrypt, along with Verifier and Holder computation in the OPRF.}
    \label{fig:perclaim_stats}
\end{figure*}


\subsection{Core claim-level results}

Since they do not depend on the total number of claims $n$, all claim-wise operations are aggregated in Figure~\ref{fig:perclaim_stats}, which also reports their average (\textsf{avg}), maximum (\textsf{max}), and outlier fraction (\textsf{\#out}).
The cost of hashing and verifying a claim is negligible, with average runtimes below $0.02$ ms. Both tasks reduce to computing a SHA-3 hash over the claim value combined with its nonce, hence they remain essentially constant across runs.
By contrast, encryption (cf.\ ll.8--11 from Algorithm~\ref{alg:create-vp}) and decryption (cf.\ ll.11--14 from Algorithm~\ref{alg:verify-vp-oprf}) dominate per-claim processing. On average, encryption performed by the Holder requires $7.6$ ms due to the combination of hashing and elliptic-curve operations needed to derive the symmetric key. Decryption on the Verifier's side is instead slightly faster, averaging $6.3$ ms.
The OPRF evaluation phase shows timings of the same order of magnitude: generating an oblivious request (cf.\ ll.8--9 from Algorithm~\ref{alg:verify-vp-oprf}) takes $7.5$ ms on the Verifier’s side, while producing the corresponding response on the Holder's side (cf.\ l.10 from Algorithm~\ref{alg:verify-vp-oprf}) averages $6.1$ ms.


\subsection{Run-level results}

\begin{figure*}[tb]
    \centering
    \includegraphics[width=\dxplot\linewidth]{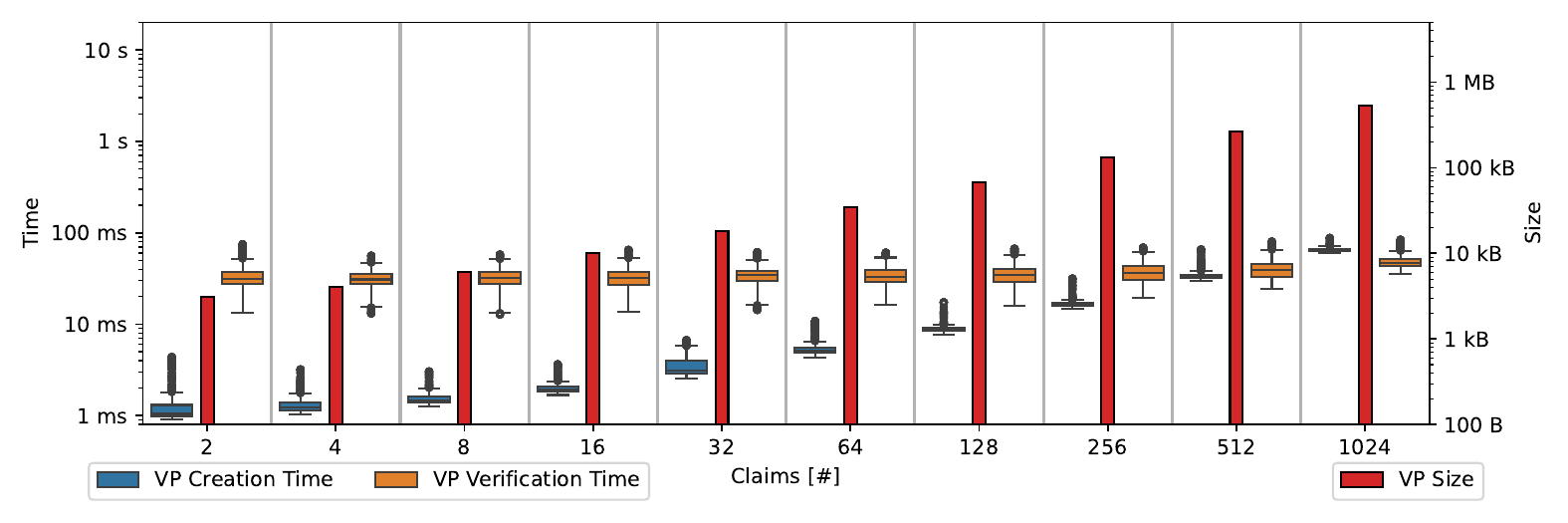}
    \caption{Verifiable presentation (\VP) statistics w.r.t.\ the claim size of their single \VC. Timings (left scale) of creation and verification are reported as boxplot, while \VP size (right scale) is reported as a bar plot since negligible variance is present.}
    \label{fig:VPres}
\end{figure*}

\begin{figure*}[tb]
    \centering
    \includegraphics[width=\dxplot\linewidth]{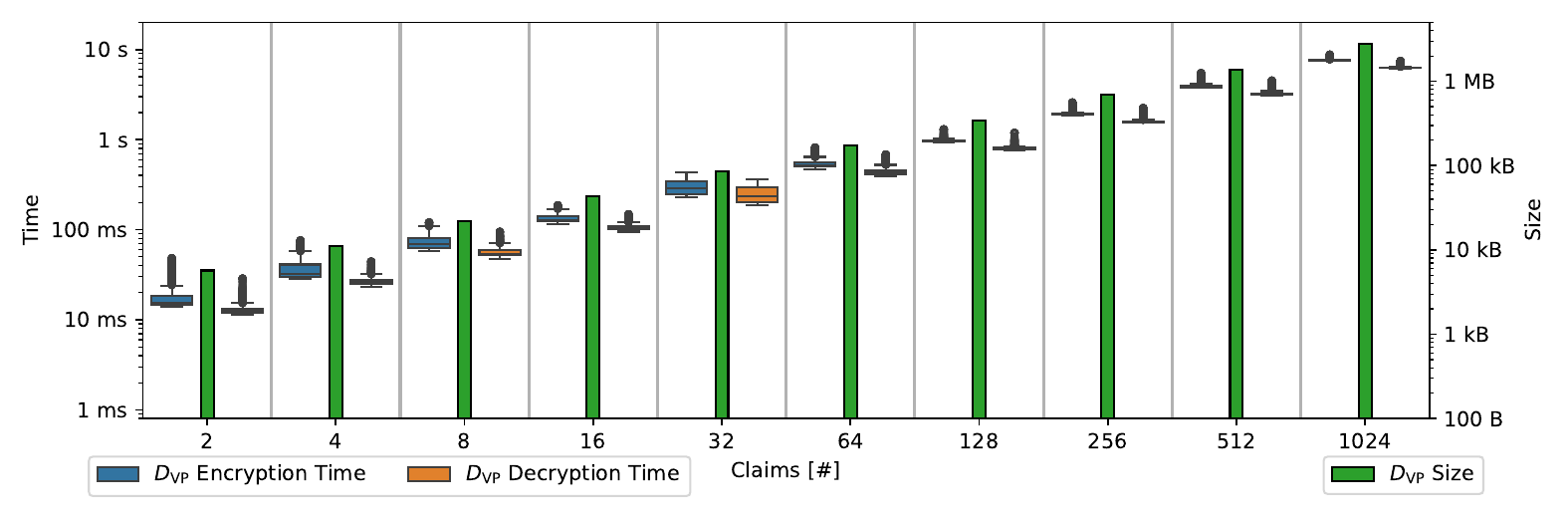}
    \caption{
        Verifiable presentation overhead caused by data $D_\VP$ creation w.r.t.\ the claim size of their single \VC.
        Timings (left scale) of encryption and decryption are reported as boxplot, while $D_\VP$ size (right scale) is reported as a bar plot since negligible variance is present.
    }
    \label{fig:DVPres}
\end{figure*}

A comparative analysis of the performance and size metrics required to present a \VC with varying numbers of claims is presented in Figure~\ref{fig:VPres}, where timings (left scale) of the VP signature creation and verification are reported as boxplot, while sizes
(right scale) of the \VP are reported as bar plots. The \VP creation time measures the time consumed by the Holder to sign the \VP and time increases with the number of claims in the \VC.
Similarly, the \VP verification time also depends on the number of claims in the presented \VC and it is comparable to the \VC verification time because it involves the same operations (\ie, DID resolution and signature verification). In every experimental configuration, the time required for both VP creation and verification is less than 100 ms. Instead, the size of the \VP ranges from 3 kB (for \VC consisting of 2 claims) to 0.53 MB (for \VC consisting of 1024 claims), and it is comparable to the size of the embedded \VC (see Figure~\ref{fig:VC_stats}). It is worth noting that VP signature creation and verification times are common across all selective disclosure approaches, as they correspond to mandatory operations defined by the W3C specification \cite{vcw3c}.

Figure~\ref{fig:DVPres} shows the additional computational overhead introduced for both the Holder and the Verifier during \VP creation and \VP verification when using the proposed approach. In particular, the plot shows the time required by the Holder to generate a symmetric key for each claim and to encrypt the claim value and the nonce with the corresponding symmetric key ($D_{VP}$ Encryption Time), and the time required by the Verifier to derive the symmetric key and decrypt all the claim values ($D_{VP}$ Decryption Time). 
The $D_{VP}$ Encryption/Decryption Time increases with the number of claims, ranging from 10 ms (for \VC consisting of 2 claims) to 10 s (for \VC consisting of 1024 claims), and the $D_{VP}$ Decryption Time is generally lower than $D_{VP}$ Encryption Time. However, the two measures remain broadly comparable in terms of computational cost as both involve elliptic curve, hashing, and encryption/decryption operations.
Here it is important to notice we present the timings for the decryption of all the claims within the \VP, a setting suitable to assess the scaling properties of the benchmark which does not represent the actual usage for selective disclosure.

Figure~\ref{fig:DVPres} shows the overhead required to store the \VP metadata generated during the Encryption by the Holder. The majority of the data generated during the \VP creation process are related to the \VP metadata. Specifically, the size of the \VP metadata is higher than the size of the corresponding \VP as it includes, for each claim, the encrypted claim's value, the corresponding initialisation vector (\IV), and the authentication tag.

\begin{figure*}[tb]
    \centering
    \includegraphics[width=\dxplot\linewidth]{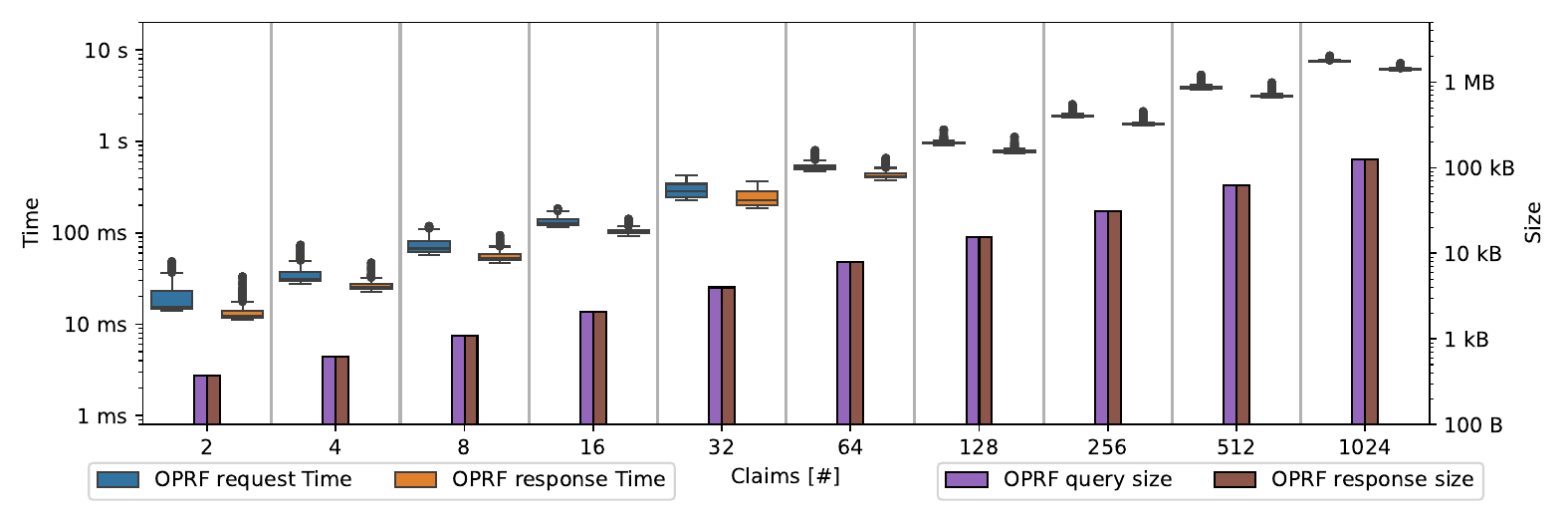}
    \caption{
        Statistics of the OPRF interaction (overhead, cf. Figure~\ref{fig:hash_stats}) w.r.t.\ the claim size of their single \VC.
        Timings (left scale) of Verifier request and Holder response are reported as boxplot, while corresponding message sizes (right scale) are reported as bar plots since negligible variance is present.
    }
    \label{fig:oprf_stats}
\end{figure*}

\begin{figure*}[tb]
    \centering
    \includegraphics[width=\dxplot\linewidth]{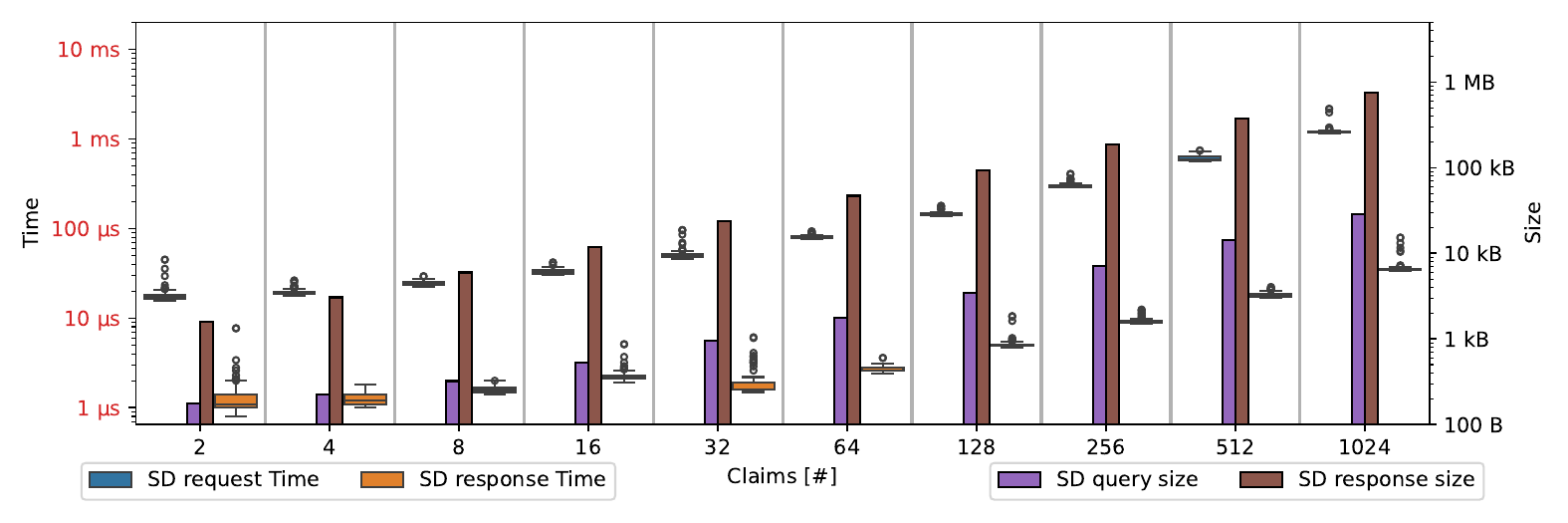}
    \caption{
    Statistics of the hash-based selective disclosure interaction (reference) w.r.t.\ the claim size of their single \VC. 
    Timings (left scale) of Verifier request and Holder response are reported as boxplot, while corresponding message sizes (right scale) are reported as bar plots since negligible variance is present.
    Time scale (left, highlighted in red) is not consistent with other plots for readability.
    }
    \label{fig:hash_stats}
\end{figure*}

After the \VP and its associated metadata are exchanged from the Holder to the Verifier, the proposed \CODSSI selective disclosure approach requires both parties to engage in one or more interactions to obliviously exchange the cryptographic material necessary for revealing the plaintext values of the claims. Figure~\ref{fig:oprf_stats} shows the time required by the Verifier to initiate a request (OPRF request time) and by the Holder to generate the corresponding response (OPRF response time) within the proposed protocol. The proposed experimental setup reflects a worst-case scenario where the Verifier requests the cryptographic material needed to decrypt all claims contained in the \VC.
The execution time for both the Holder and the Verifier increases with the number of claims in the \VC. Specifically, in the proposed approach, the time required by the Verifier to create a request is higher than the time required by the Holder to create the corresponding response. Indeed, the Verifier performs more operations (\ie, hashing and scalar multiplication) compared to the Holder.
Here it is important to notice that the choice of adaptive vs.\ batch modes does not influence the total computational and spatial time, but rather the number of interactions used to query all the keys (and consequently the number of messages that need to be exchanged).

Figure~\ref{fig:oprf_stats} also shows the amount of data that need to be exchanged between the Verifier and the Holder to obliviously request the disclosure of all claims in the \VP (OPRF query size and OPRF response size).
The amount of data generated for processing both the request and response is comparable and increases with the number of claims in the \VC.
Specifically, the request and the response size in the case of \VC consisting of 1024 claims is about 0.12 MB.
We compared these results with those obtained using a classic selective disclosure approach without obliviousness features as defined in \cite{SalveLMR22}.
In particular, Figure~\ref{fig:hash_stats} shows the time required by the Verifier to initiate a selective disclosure request for all the claims in the \VC (SD request time) and by the Holder to generate the corresponding response (SD response time), as well as the amount of data that need to be exchanged in each interaction (SD query size and SD response size).
The time required by the Holder and the Verifier to request the disclosure of claims' values is very low as it involves only exchange of nonce and plaintext claim values.
In contrast, the proposed approach requires elliptic curve operations to protect the privacy of the Verifier, resulting in higher computation overhead.
Similarly, the amount of data generated by the Verifier to initiate the request (SD query size) is about half of that required by the proposed approach (OPRF query size), as it does not include oblivious protection, and the Verifier only needs to specify the claims to be accessed.
Instead, the plot reveals that the amount of data generated by the Holder to create the response (SD response size) is higher than the size required by the proposed approach (OPRF response size).
This difference arises from the fact that the SD response includes the claim values and their associated nonces, while the proposed approach requires exchanging the results of elliptic curve scalar operations, which are then used to derive cryptographic keys and decrypt the claims.

\section{Conclusion}
\label{sec:conclusions}

This paper presented \CODSSI, a solution for preserving Verifier privacy in the process of \SSI \VP disclosure and verification.
To the best of our knowledge, this is the first solution enabling Verifiers to obtain a bounded number of claims while keeping private which subset of claims was actually requested.

By leveraging Oblivious Pseudorandom Functions, \CODSSI enables Holders to prepare Verifiable Presentations containing a large number of claims (say $N$) while allowing Verifiers to select and retrieve at most $N_o \leq N$ claims, without revealing the indices of the disclosed items. This property is achieved without compromising the Holder's ability to enforce the disclosure bound, and without requiring any expensive cryptographic assumptions beyond well-studied OPRF constructions.

\CODSSI applies naturally to several realistic domains, including financial risk scoring and medical data access, where Verifier privacy is crucial.
We formalised the security guarantees of the framework, implemented a prototype, and conducted a comprehensive experimental evaluation.
The results show that our approach is efficient and practical: disclosure times range from a few milliseconds to under two seconds for presentations containing up to 128 claims, and storage requirements remain between a few kilobytes and roughly 1 MB for the same setting, values that do not materially affect user experience in typical \SSI scenarios.

While filling a clear gap in the literature, our solution also opens important lines for further research.
The primary open challenge consists in providing Verifier-privacy guarantees against malicious Holders:
we outlined high-level directions and sketched possible solutions based on zero-knowledge proofs, but developing a fully sound and efficient construction remains an open problem.
Additionally, extending \CODSSI to alternative selective-disclosure paradigms --including \eg encrypted credential schemes-- might reduce the computational burden on the Holder by shifting parts of the \VP preparation process to the Issuer, and constitutes another promising line of investigation.
A further direction is the integration of publicly auditable \VP specifications, where a supervisory authority certifies the admissible claim set and the number of claims required by the Verifier's evaluation circuit, enabling regulatory compliance verification without compromising the Verifier's internal policy.
Finally, formalising Verifier-Privacy-Preserving Selective Disclosure as a construction-independent security notion (cf.\ Remark~\ref{rem:dual-guarantee}), defined via indistinguishability games in the standard PPT model, would strengthen the theoretical foundations and facilitate the analysis of alternative instantiations.

In conclusion, \CODSSI bridges a missing component in the \SSI ecosystem by enabling Verifier-private selective disclosure of claims.
The security of our framework is formally proved, and our prototype and experimental results confirm the feasibility and robustness of the approach, laying the groundwork for more expressive and robust privacy-preserving credential ecosystems.


\section*{Acknowledgements}

\noindent This work was partially supported by project SERICS (PE00000014) under the MUR National Recovery and Resilience Plan funded by the European Union - NextGenerationEU.

\bigskip

\noindent The authors have no competing interests to declare that are relevant to the content of this article.


\bibliographystyle{elsarticle-num}
\bibliography{arXiv_bibliography}


\appendix

\section{Notation}
\label{app:nomenclature}

For the reader's convenience, we provide in Table~\ref{tab:symbols} a list of all symbols adopted within the manuscript.

\begin{table*}[t]
    \centering
    \caption{Table of symbols, divided by primitive}
    \label{tab:symbols}
    \begin{tabular}{ccl}
        \toprule
        \textbf{Variable} & \textbf{Definition} & \textbf{Meaning}\\
        \midrule\midrule
        $\DID_\mathsf{name}$ & \texttt{did:<method>:<id>} & A \DID controlled by $\mathsf{name}$ \\
        $\sk_\mathsf{name}$ & -- & Secret key associated with $\DID_\mathsf{name}$\\
        \VC & $(C_\VC[C_\VC^h], M_\VC, P_\VC)$ & A [selectively disclosable] Verifiable Credential \\
        $M_\VC$ & -- & The set of a \VC metadata\\
        $P_\VC$ & -- & The proof of a \VC (Issuer signature)\\
        $C_\VC$ & $\{c_i : v_i\}$ & The set of claims in a \VC\\
        $c_i$ & \texttt{string} & A claim name / specialised path\\
        $v_i$ & \texttt{string} & A claim value\\
        $C_\VC^h$ & $\{c_i : x_i\}$ & The set of claim digests in a \VC\\
        \midrule
        $D_\VC$ & $C_\VC$ & The Credential Data, containing the plain claims\\
        $x_i$ & $H(v_i \| t_i)$ & The digest of claim $v_i$\\
        $H$ & $\FF_2^\star \to \XXX$ & The selective disclose ROM (SHA-3)\\
        $\XXX$ & -- & Digest space (and Input space of the (O)PRF)\\
        $t_i$ & \texttt{nonce} & The token for the digest evaluation\\
        \midrule
        \VP & $(\VC[\cdot], M_\VP, P_\VP)$ & A Verifiable Presentation \\
        $\VC[\cdot]$ & -- & A list of \VCs (referring to the same \DID)\\
        $M_\VP$ & -- & The set of a \VP metadata\\
        $P_\VP$ & -- & The proof of a \VP (Holder signature)\\ 
        \midrule
        $F$, $F_\msk$ & $\KKK \times \XXX \to \YYY$ & The (O)PRF for KDF (2HashDH), the OPRF with key $\msk$\\
        $\KKK$ & -- & The key space of the (O)PRF\\
        $\YYY$ & -- & The output space of the (O)PRF (and key space for \Enc)\\
        $\msk_\VC$ & $\in \KKK$ & Master secret key of the \VC to be used in $F$\\
        $H_1$ & $\XXX \to G$ & The 2HashDH hash function mapping elements to $G$\\
        $H_2$ & $\XXX \times G \to \YYY$ & The 2HashDH hash function finalising the OPRF output\\
        $a, b, c, r$ & -- & Internal values of 2HashDH scheme\\
        \midrule
        $\Enc_K$ & $(\IV, M, A) \to (C, T)$ & AEAD symmetric encryption (AES-GCM) with key $K$\\
        $\Dec_K$ & $(\IV, C, T, A) \to (M \lor \bot)$ & AEAD symmetric decryption (AES-GCM) with key $K$\\
        $K$ & -- & Key for \Enc, practically $k_i \in \YYY$\\
        $\IV$ & \texttt{nonce} & Random nonce for \Enc, practically $\IV_i$\\
        $M$ & -- & Plaintext for \Enc, practically $v_i$\\
        $A$ & -- & Associated data for \Enc, practically $x_i$ (or None)\\
        $C$ & -- & Ciphertext output of \Enc, practically $y_i$, included in $E_\VC$\\
        $T$ & -- & Authentication tag of $C$, practically $u_i$, included in $E_\VC$\\
        \midrule
        $D_\VP$ & $E_\VC[\cdot]$ & The \VP associated data (encrypted material)\\
        $E_\VC$ & $\{c_i : (\IV_i, y_i, u_i)\}$ & The encrypted claims of a \VC\\
        $(y, u)$ & $=\Enc_{k}(\IV, v\|t, x)$ & The encryption of claim $v$\\
        \midrule
        $\AAA$ & -- & Adversary\\
        $\IIII, \HHHH, \VVVV$ & -- & Issuer, Holder, Verifier\\
        \bottomrule
    \end{tabular}
\end{table*}

\section{Verifiable Credential Statistics}
\label{app:results}
The performance metrics related to the processing of a \VC, which are equal to the classical selective disclosure ones, are presented in Figure \ref{fig:VC_stats}, where timing distributions (boxplots, left axis) are coupled with storage (bar plots, right axis).
Both the average creation and the verification time of the \VC increase with the number of claims. In particular, the computational cost required by the \VC creation grows more rapidly with the number of claims compared to the verification.
This difference is primarily due to the type of operations executed for \VC creation and verification: the \VC creation involves mainly hashing and digital signature operation on each claim, while \VC verification involves DID resolution (see Sec.~\ref{ssec:ssi}) and signature verification.
The size required by both the credential and its associated metadata increases linearly with the number of claims, as shown by the bar plot in Fig.~\ref{fig:VC_stats}. In particular, the plot reveals that the size of the credentials is slightly smaller than the size of its metadata, \eg, a \VC consisting of 1024 claims takes 0.40 MB
while the corresponding \VC metadata takes 0.78 MB. Indeed, the \VC consists of the hashed claims' values, while the \VC metadata includes the claims' values and the nonces. 
\begin{figure*}[tb]
    \centering
    \includegraphics[width=\dxplot\linewidth]{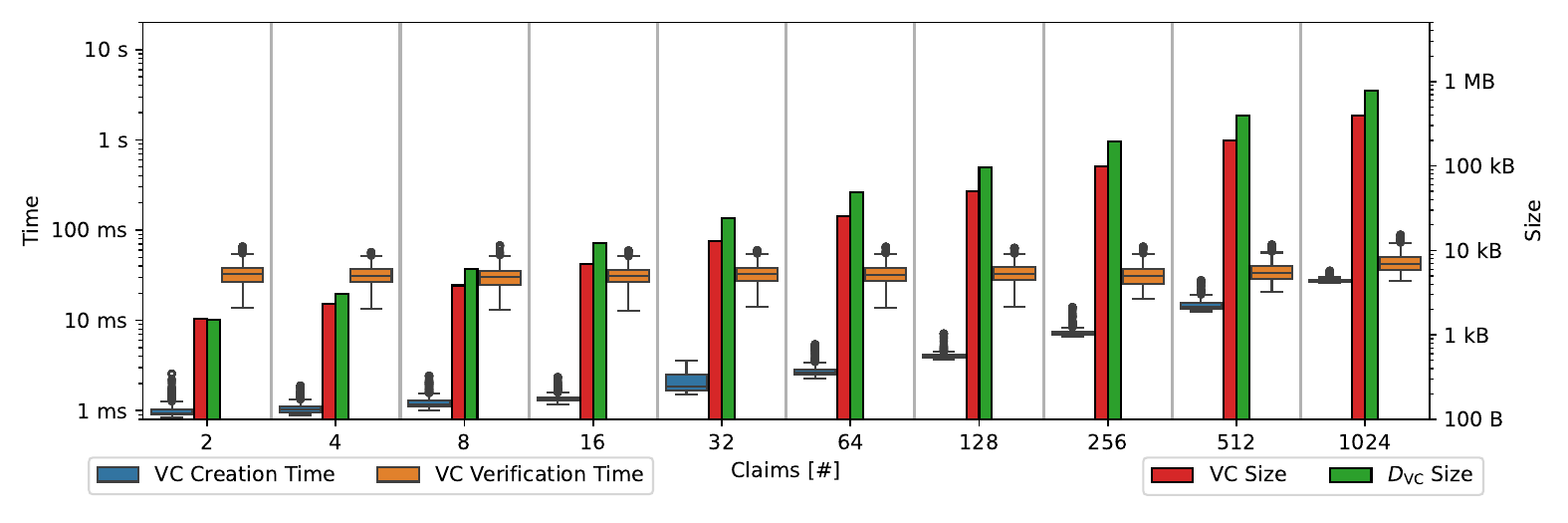}
    \caption{Verifiable credential (\VC) statistics w.r.t.\ their claim size. Timings (left scale) of creation and verification are reported as boxplot, while sizes (right scale) of \VC and $D_\VC$ are reported as bar plots since variance is negligible.}
    \label{fig:VC_stats}
\end{figure*}

\section{Formal security derivation}
\label{app:security}

In this appendix, we present a precise threat model focusing on the Holder--Verifier interaction, which is the locus of the privacy trade-offs introduced by our construction, and we show that \CODSSI guarantees security under the introduced threat model.

\subsection{Threat model and adversarial capabilities.}
\label{ssec:threat-model}
We consider adversaries that may corrupt any of the three roles, but our main focus is on the Holder--Verifier interaction; the Issuer is treated as an honest(-but-curious) party for the purposes of the \VP creation protocol (cf.\ Algorithm~\ref{alg:create-vp}).

\begin{description}
  \item[Malicious Verifier (\(\AAA_\VVVV\)).]  
    The Verifier may act arbitrarily: she can deviate from the OPRF client specification, make adaptive or batched queries to the Holder, attempt to infer non-disclosed claims from protocol transcripts, collude with external parties, or mount replay/chosen-input strategies.  
    The Verifier's main goal consists of learning the values of claims for which she has not obtained the decryption keys.

  \item[Malicious Holder (\(\AAA_\HHHH\)).]  
    The Holder may behave arbitrarily as the OPRF server (\eg, deviate from the protocol, attempt to inject malformed responses, or refuse to respond) yet must behave compliant in the \VP creation phase. The Holder's main goal consists of learning which claim indices the Verifier queried (\ie, breaking client privacy).

  \item[Malicious Issuer (\(\AAA_\IIII\)).]  
    Although the Issuer does not participate in the oblivious protocol, we consider the standard issuance threats: a malicious Issuer could issue malformed or maliciously constructed credentials. We therefore assume that Holders (and later the verifiers) perform the usual signature verification and consistency checks on $\VC$ and $C_\VC^h$ prior to participating in disclosure.

  \item[External malicious adversary.]  
    An active adversary may observe, modify, or reorder messages between Holder and Verifier. We require the OPRF channel to provide standard transport security (or that the protocol itself is implemented to resists active interference); encryption outputs $(\IV_i,y_i,u_i)$ are assumed public but integrity-protected by AEAD and parties are assumed mutual authenticated (\eg under their \DIDs' public key).
\end{description}

\subsection{Security goals and assumptions}
\label{ssec:security-goals}
Under the preceding adversarial model, we recall the following properties required under our security requirements (cf.\ Theorem~\ref{thm:security}):
\begin{enumerate}
  \item \textbf{Confidentiality of non-disclosed claims:} an adversary that does not obtain the correct claim key $k_i$ learns no information on $v_i$ beyond trivial leakage;
  \item \textbf{Integrity and authenticity:} any modification of encrypted credential data is detected and invalidated by the Verifier prior to release of plaintext;
  \item \textbf{Client (Verifier) obliviousness:} a malicious Holder (server) cannot distinguish which claim indices the Verifier queried during the OPRF protocol beyond what is implied by the Holder's policy-triggered actions (\eg, denying service after a quota is reached);
  \item \textbf{Holder policy enforcement (optional):} the Holder may enforce constraints on disclosure (\eg, a quota on the number of keys issued, rate limits, or mandatory interactive checks); such enforcement must not compromise obliviousness more than the policy's observable effects;
  \item \textbf{Unforgeability of disclosures:} a malicious Holder cannot cause the Verifier to accept a plaintext claim that is not derivable from the original credential bindings and the Issuer's signatures.
\end{enumerate}

We also recall the following standard assumptions over our primitives (cf.\ Theorem~\ref{thm:security}):
\begin{itemize}
  \item The OPRF implementation is UC-secure (or at least secure in the ROM) against malicious adversaries (we adopt the classical 2HashDH model from Section~\ref{ssec:oprf}, yet any instantiation with equivalent guarantees is suitable).
  \item The symmetric encryption primitive used to instantiate $\Enc_{k_i}$ is an AEAD scheme providing IND-CPA confidentiality and INT-CTXT authenticity (we adopt AES-GCM from Section~\ref{ssec:aes-gcm}); tags and nonces may be transmitted in the clear alongside ciphertexts.
  \item Credential selective-disclosure primitives (commitments of SHA-3, in our implementation) satisfy the usual unforgeability and privacy properties, \ie the Random Oracle Model (ROM) is assumed.
  \item The Holder correctly enforces any policy decisions (quotas, revocation checks) in a manner that is independent from the specific claim indices except for the observable policy outputs (\eg, ``reject after quota exhausted'').
\end{itemize}


\subsection{Security theorems}
\label{ssec:securityTheorems}

We now formalise and prove the properties listed in Section~\ref{ssec:security-goals} by reduction to the security of the underlying primitives. 
Let $\lambda$ denote the security parameter and let $\negl(\cdot)$ denote negligible functions in $\lambda$.
In what follows, we propose ``proof sketches'' intended to expose the reduction strategy; full reductions follow standard composition arguments.

\begin{theorem}[Claim confidentiality]
\label{thm:confidentiality}
Let $F$ be a UC-secure OPRF (modelled as a PRF in the ROM), $H$ be a secure derivation in the ROM, and $\Enc$ be an IND-CPA secure AEAD scheme.
Consider a verifiable presentation of $N$ claims $\mathsf{pub} = {(x_i, \IV_i, y_i, u_i)}_{i=1}^N$, where $y_i = \Enc_{k_i}( \IV_i,m_i,u_i)$, $k_i = F_{\msk_\VP}(x_i)$, and $x_i = H(m_i)$ and where $m_i$ is here a randomised claim value $m_i = v_i \| t_i$.

Let $\AAA_\VVVV$ be any PPT adversary that obtains the OPRF outputs for at most $q$ indices $Q\subseteq \{1, \dots, N\}$ (\ie it queries the OPRF on $Q$, $|Q|\le q$).  
Then, for every index $j\notin Q$ the ciphertext triple $(\IV_j,y_j,u_j)$ is computationally indistinguishable from an AEAD encryption of any other equal-length message under a fresh, uniform key independent of all other keys.  
Equivalently, the distinguishing advantage of $\AAA_\VVVV$ on any such unqueried index $j$ is negligible; more concretely there exist reductions $\BBB_{\mathsf{AEAD}},\CCC_{\mathsf{OPRF}}, \DDD_{\mathsf{ROM}}$ such that
\[
\mathrm{Adv}^{\mathsf{conf}}_{\AAA_\VVVV}(j)\;\le\;
\mathrm{Adv}^{\mathsf{IND\text{-}CPA}}_{\BBB_{\mathsf{AEAD}}}
\;+\;\mathrm{Adv}^{\mathsf{OPRF}}_{\CCC}
\;+\;\mathrm{Adv}_\DDD^{\mathsf{ROM}}
\;+\;\negl(\lambda).
\]

\end{theorem}

\begin{proof}
Fix an unqueried index \(j\notin Q\). We argue that any non-negligible distinguishing advantage of \(\AAA_\VVVV\) on \((\IV_j,y_j,u_j)\) implies a break of one of the stated primitives.

\textbf{(i)} Because \(x_j=H(m_j)\) and \(m_j\) is kept secret (it is encrypted), the ROM modelling of \(H\) allows us to treat \(x_j\) as a value that is \emph{indistinguishable from random} to any adversary that does not know \(m_j\), except to the extent the adversary can (i) find a preimage \(m'\) with \(H(m')=x_j\) by querying the random oracle, or (ii) otherwise exploit oracle programmability. Any non-negligible advantage due to such manipulations is captured by \(\mathrm{Adv}^{\mathsf{ROM}}_{\DDD}\), the advantage of an adversary \(\DDD\) that violates the intended ROM behaviour (preimage-finding or distinguishing programmed points). Thus, modulo \(\mathrm{Adv}^{\mathsf{ROM}}_{\DDD}\), the input \(x_j\) to the OPRF can be treated as a random (from the adversary's viewpoint) value.

\textbf{(ii)} By the PRF-like (UC) security of the OPRF, if the adversary has not queried \(x_j\) to the OPRF oracle, then the output \(k_j = F_{\msk_\VP}(x_j)\) is computationally indistinguishable from a uniformly random key in the AEAD key space (and different unqueried inputs give independent outputs up to negligible advantage). Combining this with \textbf{(i)} (\ie, treating \(x_j\) as unpredictable except for the ROM term), we conclude that from the adversary's viewpoint \(k_j\) behaves like a fresh uniform key, except with probability bounded by \(\mathrm{Adv}^{\mathsf{OPRF}}_{\CCC}+\mathrm{Adv}^{\mathsf{ROM}}_{\DDD}\).

\textbf{(iii)} Fix any equal-length message \(m^\star\). Under a uniform unknown key \(k_j\), IND-CPA security of the AEAD scheme \(\Enc\) implies that the ciphertexts \(\Enc_{k_j}(m_j, \IV_j)\) and \(\Enc_{k_j}(m^\star;\IV_j)\) are computationally indistinguishable, except with advantage bounded by \(\mathrm{Adv}^{\mathsf{IND\text{-}CPA}}_{\BBB_{\mathsf{AEAD}}} + \mathrm{Adv}^{\mathsf{OPRF}}_{\CCC} + \mathrm{Adv}^{\mathsf{ROM}}_{\DDD}\).

If \(\AAA_\VVVV\) distinguishes \((\IV_j,y_j,u_j)\) from an encryption of some other equal-length message with non-negligible advantage, then one of the following must hold:
\begin{enumerate}
  \item \(\AAA_\VVVV\)'s advantage yields a non-negligible IND-CPA break against \(\Enc\) (construct \(\BBB_{\mathsf{AEAD}}\) that forwards the AEAD challenge as the \(j\)-ciphertext while simulating other indices using a simulated \(\widetilde{\msk}_\VP\)); or
  \item \(\AAA_\VVVV\) has learned non-negligible information about \(k_j\) (despite not querying the OPRF on \(x_j\)), which yields an OPRF break (construct \(\CCC\) that uses \(\AAA_\VVVV\)'s outputs to distinguish OPRF from random); or
  \item \(\AAA_\VVVV\) exploited ROM-specific behaviour (\eg, found a preimage or distinguished programmed points of \(H\)), which yields an attack captured by \(\mathrm{Adv}^{\mathsf{ROM}}_{\DDD}\).
\end{enumerate}
Therefore, the adversary's distinguishing advantage is upper-bounded by the sum of the three primitive advantages plus negligible simulation losses, giving the claimed bound.
\end{proof}

\begin{theorem}[Integrity and authenticity]
\label{thm:integrity}
Assume the AEAD scheme $\Enc$ provides INT-CTXT authenticity. Then any PPT adversary that modifies an encrypted claim tuple $(\IV_i,y_i,u_i)$ and causes the Verifier to accept a forged plaintext (or malformed authentication tag) with non-negligible probability would produce an existential forgery against the AEAD scheme. 
Hence, such forging probability is negligible under the INT-CTXT assumption.
\end{theorem}

\begin{proof}
Any adversary that leads the Verifier to accept a tampered ciphertext (\ie, the AEAD verification passes and a different plaintext is produced) can be converted into an adversary that forges a ciphertext under the AEAD challenger. 
Since our protocol describes no particular interactions with the \Enc function, the reduction simply forwards the challenger-supplied ciphertexts as the protocol ciphertexts and uses the adversary's successful acceptance as a forgery. 
Thus, integrity follows directly from AEAD INT-CTXT security.
\end{proof}

\begin{theorem}[Verifier obliviousness]
\label{thm:obliviousness}
Assume the OPRF protocol $F$ is UC-secure (or secure in the ROM with the standard OPRF privacy notion). 
Then, for any PPT adversary $\AAA_\HHHH$ (malicious Holder acting as OPRF server), the Holder's view on the Verifier indices choices is computationally indistinguishable when the Verifier queries any two index sets $I_0, I_1$, except for leakage that is explicitly caused by the query set size.
Formally, the Holder's distinguishing advantage is bounded by the advantage of breaking OPRF client privacy plus negligible terms.
\end{theorem}

\begin{proof}
The proof directly follows from the security of the OPRF, being used as-is in the protocol.
Furthermore, by the UC-security (or standard OPRF privacy) guarantee, interactions between an honest client and an honest server are simulatable: there exists a simulator that produces server-side views indistinguishable from real protocol runs while only learning the number of queries but not their specific indices. 
Thus, any Holder adversary that distinguishes between two different equally-sized sets of the Verifier's queries could be used to break the OPRF privacy property. 
Policy-triggered observable events (\ie, the server denying further queries after $q_{\max}$ responses) are considered allowable leakage and are modelled explicitly; the indistinguishability statement is modulo this leakage.
\end{proof}

\begin{theorem}[Unforgeability of disclosures and Issuer verifiability]
\label{thm:unforgeability}
Assume the credential selective-disclosure primitives are unforgeable. 
Then, a malicious Holder cannot make a Verifier accept a plaintext claim as valid unless that claim is consistent with some legitimately issued credential $\VC$ and the Issuer's signatures.
\end{theorem}

\begin{proof}[Proof sketch]
Acceptance of a claim by the Verifier requires two conditions: (i) the decrypted plaintext must pass AEAD verification (guaranteed by Theorem~\ref{thm:integrity}), and (ii) the decrypted claim must be consistent with the Issuer-bound commitments and proofs contained in the \VP (this is the selective-disclosure verification step). 
If a Holder forges a claim that is accepted, then either it has forged an AEAD ciphertext (contradicting AEAD integrity) or it has produced a credential/commitment pair and accompanying proof that passes Issuer verification despite not being issued -- actually contradicting the unforgeability/soundness of the selective-disclosure primitive. 
Both cases reduce to standard forging/soundness breaks, hence are negligible under the assumptions.
\end{proof}

\begin{corollary}[Compositional security]
Under the stated assumptions (OPRF security, AEAD IND-CPA+INT-CTXT, and selective-disclosure unforgeability / soundness in the ROM), the construction satisfies the informal security goals in Section~\ref{ssec:security-goals} (confidentiality of non-disclosed claims, integrity / authenticity, Verifier obliviousness up to policy leakage, Holder policy enforcement, and unforgeability). 
More precisely, for any PPT adversary $\AAA$, there exist reductions that bound the adversary advantage by the sum of the primitive advantages plus negligible terms.
\end{corollary}

\begin{proof}[Sketch]
Immediate by combining Theorems~\ref{thm:confidentiality}--\ref{thm:unforgeability} and observing that the protocol composes the primitives in a black-box fashion (OPRF for key derivation, AEAD for encryption, and selective-disclosure proofs for Issuer binding).
\end{proof}

\end{document}